\documentclass[11pt,letterpaper]{article}
\usepackage{amsfonts,amsmath,amsthm,amssymb,dsfont, etex}

\usepackage{tikz}
\usetikzlibrary{arrows,automata}
\usepackage{comment, times}
\usepackage[margin=1in]{geometry}
\usepackage{soul,color}
\usepackage{graphicx,float,wrapfig}
\usepackage{mathrsfs}
\usepackage[usenames,dvipsnames]{pstricks}
\usepackage{epsfig}
\usepackage{pst-grad} % Fdi gradients
\usepackage{pst-plot} % Fdi axes
\usepackage{mathdots}
\usepackage[linesnumbered,boxed,ruled,vlined]{algorithm2e}
\newcommand{\distr}{\mathcal{U}}
\newcommand{\distrD}{\mathcal{D}}
\newcommand{\distrP}{\mathcal{P}}
\newcommand{\Prand}{\distrP_{\textsf{rand}}}
\newcommand{\Ex}{\operatorname*{\mathbb{E}}}

\usepackage{url}

\usepackage{tabu}

\usepackage[normalem]{ulem}

\newtheorem{theo}{Theorem}[section]

\newtheorem{lemma}[theo]{Lemma}

\newtheorem{claim}{Claim}

\newtheorem{cor}[theo]{Corollary}
\theoremstyle{definition}
\newtheorem{defi}[theo]{Definition}
\theoremstyle{plain}
\newtheorem{rem}[theo]{Remark}

\newtheorem{ft}[theo]{Fact}

\newenvironment{proofof}[1]{\begin{proof}[Proof of #1]}{\end{proof}}

\newenvironment{reminder}[1]{\bigskip
	\noindent {\bf Reminder of #1  }\em}{\smallskip}

\renewcommand{\epsilon}{\varepsilon}
\newcommand{\WH}{\widehat}
\newcommand{\OL}{\overline}
\newcommand{\odistr}{\OL{\distr}}
\newcommand{\BCAST}{\textsf{BCAST}}
\newcommand{\event}{\mathcal{E}}
\newcommand{\Ffullrank}{F_{\textsf{full-rank}}}

\newcommand{\distrA}{\mathcal{A}}
\newcommand{\Arand}{\distrA_{\sf rand}}
\newcommand{\distrS}{\mathcal{S}}

\newcommand{\polylog}{\operatorname{polylog}}

\def\ShowAuthNotes{1}
\ifnum\ShowAuthNotes=1
\newcommand{\authnote}[2]{\ \\ \textcolor{red}{\parbox{0.9\linewidth}{[{\footnotesize {\bf #1:} { {#2}}}]}}\newline}
\else
\newcommand{\authnote}[2]{}
\fi

\date{}
\title{Broadcast Congested Clique: Planted Cliques and Pseudorandom Generators}

\author{Lijie Chen\thanks{Supported by NSF CCF-1741615 (CAREER: Common Links in Algorithms and Complexity). This work was done in part while the authors were visiting the Simons Institute for the Theory of Computing.}\\ MIT \and Ofer Grossman\thanks{Supported by the Fannie and John Hertz Foundation Fellowship, an NSF Graduate Research Fellowship, NSF CNS-1413920,
DARPA 491512803, 
Sloan 996698, and 
MIT/IBM W1771646.} \\ MIT}

\begin{document}
	\maketitle
    
\begin{abstract}
\normalsize
%We construct a pseudo-random generator for the Broadcast Congested Clique. This is the first pseudo-random generator which fools a distributed system. Additionally, using the same techniques we show the first average case hardness result for the Broadcast Congested Clique.
We develop techniques to prove lower bounds for the \BCAST($\log n$) Broadcast Congested Clique model (a distributed message passing model where in each round, each processor can broadcast an $O(\log n)$-sized message to all other processors). Our techniques are built to prove bounds for natural input distributions. So far, all lower bounds for problems in the model relied on constructing specifically tailored graph families for the specific problem at hand, resulting in lower bounds for artificially constructed inputs, instead of natural input distributions.

One of our results is a lower bound for the directed planted clique problem. In this problem, an input graph is either a random directed graph (each directed edge is included with probability $1/2$), or a random graph with a planted clique of size $k$. That is, $k$ randomly chosen vertices have all of the edges between them included, and all other edges in the graph appear with probability $1/2$. The goal is to determine whether a clique exists. We show that when $k = O(n^{1/4 - \epsilon}),$ this problem requires a number of rounds polynomial in $n$.

Additionally, we construct a pseudo-random generator which fools the Broadcast Congested Clique. This allows us to show that every $k$ round randomized algorithm in which each processor uses up to $n$ random bits can be efficiently transformed into an $O(k)$-round randomized algorithm in which each processor uses only up to $O(k \log n)$ random bits, while maintaining a high success probability. The pseudo-random generator is simple to describe, computationally very cheap, and its seed size is optimal up to constant factors. However, the analysis is quite involved, and is based on the new technique for proving lower bounds in the model.

%That is, a distribution $D$ which can be sampled with few random bits such that no low-round protocol can distinguish between inputs drawn from $D$ and inputs drawn from the uniform distribution, except with low probability. The construction of the pseudo-random generator relies on proving that a certain close-to-uniform distribution of low rank matrices (in which each processor gets a row of the matrix as their pseudo-random string) is indistinguishable from random.

The technique also allows us to prove the first average case lower bound for the Broadcast Congested Clique, as well as an average-case time hierarchy. We hope our technique will lead to more lower bounds for problems such as triangle counting, APSP, MST, diameter, and more, for natural input distributions. 
\end{abstract}

\section{Introduction}

%\lnote{comments need to be addressed in the introduction}

%\lnote{- What happens with non-broadcast congested clique and similar model? What's the obstacle to generalisation?
%	- k is used for both running time and output size in introduction and appendix, sometimes even in the same paragraph. This is highly unclear and should be fixed.
%}

%Randomization is a powerful algorithmic tool for general computation, and for distributed computation in particular.

Recently there has been a surge of lower bound results in the CONGEST, and Broadcast Congested Clique models (\cite{approx, power, korhonen2017deterministic, cliqueCONGEST} to list a few, though there are many more). In general, these results take the following approach: a specific carefully chosen family of graphs is constructed, and the vertices are partitioned into two parts (or, in some rare cases, three parts), and a reduction from classical two-party communication complexity is used.

A main downside of this approach is that it merely proves worst case lower bounds. For all we know, for more natural input distributions (instead of the artificially constructed graph families used for these communication complexity reductions), many problems which are worst-case hard may become easy. One of our main goals in this paper is to develop techniques that prove lower bounds for more natural input distributions.

%So, in this work, we initiate the study of average case lower bounds for message passing distributed models. Instead of specifically constructing graph families to show lower bounds, we focus on proving lower bounds for natural problems on natural input distributions.

The distributed model we consider is the Broadcast Congested Clique model (\BCAST(1))\footnote{Note that every lower bound for \BCAST(1) can be extended to a lower bound for \BCAST($\log n$) with only a $\log n$ factor loss in the number of rounds}. In this model, there are $n$ processors, each with unlimited local computational power. Computation proceeds in rounds, and in each round each processor broadcasts a single bit\footnote{It is standard to use messages of size $O(\log n)$, but for our purposes it is more natural to consider single-bit messages. All of our results generalize to the setting of logarithmic sized messages.} to all other processors (within a single round, a processor must broadcast the same bit to all other processors).

Our results include a lower bound for the planted clique problem. Additionally, we show how to use our techniques to construct a pseudo-random generator that fools the Broadcast Congested Clique model. This is the first pseudo-random generator which fools a distributed message passing model. As simple corollaries, we obtain an average-case lower bound for the model, and an average-case time hierarchy for the model.

\subsection{Our Approach}

\paragraph{How to deal with distributions}
Consider the (directed) planted clique problem\footnote{Note that in the broadcast congested clique model, as opposed to the CONGEST or unicast congested clique model, it is not possible to reduce from directed to undirected in one round.}. In this problem, the input is a directed graph which is either a random graph (each directed edge is included with probability $1/2$), or a random graph where $k$ of the vertices are chosen at random and all edges within this set of $k$ vertices are included (these $k$ vertices are called the planted clique). The goal is to distinguish between these two distributions (or, one can consider the search version of the problem, in which a graph with a planted clique is given, and the goal is to find the clique). For now, one can think of $k$ as approximately $n^{1/4}$.

%Because a random graph's largest clique is of size $\Theta(\log n)$ (with high probability), the problem makes sense for larger values of $k$. Once $k$ goes substantially above $\sqrt{n}$, it is possible to find the clique by considering the vertices with highest degree. Hence, the interesting values of $k$ are between approximately $\log n$ and $\sqrt{n}$.

When trying to prove a lower bound for the problem, one's first approach may be to try reducing to two-party communication complexity: split the vertices into two parts, and argue that a lot of information must pass between the parts. This approach will not work for the problem, since no matter how the graph is split up, at least one of the parts will be able to detect the presence of the clique, since at least one of the parts must have many of the clique's vertices.

One's next approach may be the following: show that for any Congested Clique algorithm, after $t$ rounds, the distribution of transcripts of the Congested Clique algorithm (the ``transcript" is the history of the algorithm; that is, the ``transcript" is a list of all messages sent so far as well as who sent which message and when the message was sent) if the input were uniform is close to the distribution of transcripts if the input had a planted clique. To prove this, one may take an inductive approach: show that if the transcripts were similar for $t-1$ rounds, the next round can distinguish between the distributions with only low probability. In this round, all vertices are broadcasting, so one may try to use another approach: show that each vertex reveals little information about whether the graph has a planted clique or not, and therefore, the whole round reveals little information about whether there is a clique. This way, we would only need to analyze a single broadcast at a time, a much simpler task than considering a whole round at once.

There is an issue, however. The problem is that the inputs of different nodes may not be independent. Therefore, it is possible that while each processor's broadcast on its own will not reveal substantial information about whether the graph contains a clique, when we consider many processors' broadcasted bits the information revealed may be more than the sum of the information of the individual broadcasts.

To get around this, we instead split the planted clique distributions into a sum of many distributions, such that for each of these distributions, all processors' inputs are independent. Specifically, we can write the planted clique distribution as a sum over all possible cliques $C$ of a random graph with a clique planted at $C$. Notice that after fixing $C$, each vertex's input is independent of all other vertices' inputs. Now, when considering whether an algorithm distinguishes between a random graph with a clique at $C$ and a truly random graph, we can consider each node's output on its own, instead of trying to deal with all nodes at the same time. This is one of our main high level ideas: splitting the distribution into many distributions where in each one, different vertices' inputs are independent. This greatly simplifies the analysis by letting us avoid having to deal with many messages at once, and is what makes proving the lower bounds possible.

\paragraph{Statistical Inequalities}
The idea of partitioning a distribution into distributions with independent vertex inputs makes proving the lower bounds possible, but there is still lots of technical work to do. Specifically, we now have a bunch of distributions, and we need to show that any algorithm can distinguish only few of those distributions from uniform. At the high level, the idea to do this is to show that for any algorithm, for almost all cliques, when a vertex broadcasts a bit, the probability of that bit being broadcast with the clique vs without the clique is similar. This is basically a problem about Boolean functions: if we let $f$ be the function which takes in the node's input, and outputs what the node will broadcast, we wish to show that for almost all cliques $C$, when $f$'s input is uniform, the output distribution is similar to the distribution when $f$'s input is chosen with a planted clique at $C$. The inequalities we show are at the high level similar to this, but many more technical issues arise. For example, since we have to prove a multi-round lower bound, we have to condition on what a node broadcasted in previous rounds, so instead of proving the statistical inequalities for total functions, we instead have to prove the inequalities for functions defined on only part of $\{0, 1\}^n$. These inequalities are not true for all partial functions, but we manage to prove that they are true for all functions which are defined on a large enough subset. So then, there is another challenge of proving that with high probability, after conditioning on a transcript, the set of possible inputs to a node is large. 

\subsection{Our Results}

%\lnote{A third issue is that it is not really clear what the biggest conceptual contribution the paper makes. I understand most of the details of the "toy example," i.e., replacing one random bit by a pseudorandom bit. The overall proof strategy (e.g., top of Page 6) is clearly described. But, is this the main conceptual contribution, i.e., figuring out that an inductive proof using this strategy will work or is the main contribution in proving some of the statistical distance inequalities (e.g., Lemma 3.1)? }

%\lnote{make this section more focused?}

%\lnote{add an average case time hierarchy theorem for BCAST}

\paragraph{Lower bounds for the planted clique problem}

One of the problems we consider is the (directed) planted clique problem. In this problem, the input is a directed graph which is either a random graph (each directed edge is included with probability $1/2$), or a random graph where $k$ of the vertices are chosen and all edges within this set of $k$ vertices are added. The goal is to distinguish between these two distributions (or, one can consider the search version of the problem, in which a graph with a planted clique is given, and the goal is to find the clique). Because a random graph contains cliques of size $\Theta(\log n)$, the problem makes sense for larger values of $k$. Once $k$ goes substantially above $\sqrt{n}$, it is possible to find the clique by considering the vertices with highest degree. Hence, the interesting values of $k$ are between approximately $\log n$ and $\sqrt{n}$.

In the classical (non-distributed) setting, the planted clique problem is very well studied. There exists a spectral algorithm solving the problem when $k = O(\sqrt{n})$, and it remains a major open problem in complexity theory to understand whether the problem is hard when $k$ is smaller.

It has been shown that proving lower bounds in the unicast CONGESTED-CLIQUE\footnote{In the unicast model, in each round each processor can send one bit to each other processor without the requirement that the same bit is broadcast to all other processors; that is, within a single round a processor may choose to send the message ``1'' to some processors, and ``0'' to others.} would imply some strong circuit lower bounds \cite{power}, so finding lower bounds for the problem in this setting is quite a challenge. Finding upper bounds for planted clique in the unicast CONGESTED-CLIQUE model is an interesting problem, but it seems difficult -- maybe impossible -- to improve upon simple sampling-based algorithms for the problem. The next model one can look at is the broadcast CONGESTED-CLIQUE model, which is the model we consider. Specifically, we prove that for cliques of size $O(n^{1/4 - \varepsilon})$, the planted clique problem requires polynomially many rounds:	
\begin{theo}[Planted Clique lower bound]
	When $k = n^{1/4 - \epsilon}$ for a constant $\epsilon$, no $n^{o(1)}$ round $\BCAST(1)$ protocol $\Pi$ can distinguish between $\distrA_{\sf rand}$ and $\distrA_{k}$ with advantage\footnote{An algorithm distinguishing between two distributions $D_1$ and $D_2$ with advantage $\epsilon$ is an algorithm $A$ which, when given a random sample $s$ which with probability $1/2$ is drawn from $D_1$ and probability $1/2$ is drawn from $D_2$, then $A$ can successfully guess from which distribution $s$ was drawn with probability $1/2 + \epsilon$.} $\Omega(1)$.
\end{theo}

%The most common approach towards proving lower bounds for the CONGEST model, and the broadcast CONGESTED-CLIQUE model, is to split the graph into two parts, and then use a reduction from a communication complexity problem, and show that a lot of information must pass from one part of the graph to the other part. However, since the total amount of information that can pass between the parts is bounded by $O(n)$ per round, if we can show that $b$ bits need to be communicated, we achieve a $\Omega(n/b)$ lower bound for the problem. Notice that no such approach could work here. If we split the graph into two parts, at least one of the parts will have a clique of size at least $k/2$, and the parts don't necessarily have to communicate between them in order to distinguish the graph from random. New ideas are needed.

\paragraph{Pseudo-random generators for distributed computation:}

Historically, pseudo-random generators were used to fool adversaries modeled as Turing Machines or Circuits. One of our main contributions is constructing the first pseudo-random generator which fools a distributed message passing setting. That is, we show how within the \BCAST(1) Congested Clique model, the processors can each sample a small random seed, and end with each processor having a longer string than it started with, such that these longer strings \textit{look random} to the system. That is, no low-round \BCAST(1) protocol can distinguish between these pseudo-random strings and truly uniformly random strings within few rounds (except with some low probability). Thus, any algorithm can use a pseudo-random string instead of true randomness, thereby saving random bits.

Pseudo-random generators have been used to derandomize specific problems in message passing models, for example in \cite{meravprg}. Our construction is the first pseudo-random generator which fools \textit{all} low-round algorithms in the \BCAST  model (as opposed to just fooling a specific algorithm).

\begin{defi}[\BCAST(1) pseudo-random generator]
	A $(k,m,n,\ell)$ \BCAST(1) pseudo-random generator (PRG) is an $n$-processor \BCAST(1) protocol $\Pi$ such that:
	\begin{itemize}
		\item At the beginning every processor independently gets $k$ private uniform and independent random bits as input.
		\item After participating in protocol $\Pi$, each processor outputs $m$ bits (these bits are not broadcasted: they are the node's pseudo-random bits).
		\item The joint distribution of all processors' output bits cannot be distinguished (with better than $\frac{1}{n}$ probability\footnote{All of our constructions in the paper can achieve arbitrarily low inverse polynomial probability.}) from a truly uniform random distribution by any $\ell$ round \BCAST(1) protocol. Specifically, the statistical distance between the distribution of the protocol's transcript when using a pseudo-random generator and the distribution of the transcript when using true randomness is small.
	\end{itemize}
\end{defi}

%\lnote{Definition 1.1 :
%Does every processor really output m bits? Broadcasting them would imply Θ(m) number of rounds and that's not always O(k). To my understanding the whole point is to say that we can generate m bits from the bits provided by the processor(s), and generated string of bits for a single processor has m bit.
%}

%\lnote{justify the indistinguishibility of transcript condition}

%\lnote{Theorem 1.2 works for k=omega(logn) and then you claim that you work for any k round algorithm where k=Omega(logn). Decide which one is it.}

%\lnote{$k$ is both for rounds and seed length, need to be more careful}

\begin{theo}\label{theo:PRG-formal}
	For all $m = O(n)$ and $k = \Omega(\log n)$, there exists an $(O(k), m, n, \Theta(k))$ \BCAST(1) PRG that can be constructed within $O(k)$ rounds. In particular, the PRG works as follows
	
	\begin{itemize}
		\item Each processor gets $k + k \cdot \frac{(m-k)}{n} = O(k)$ private random bits.
		
		\item Then in $O\left(\frac{m-k}{n} \cdot k \right) = O(k)$ rounds, all processors broadcast their last $k \cdot \frac{(m-k)}{n}$ random bits. And they use that to construct a random matrix $M \in \{0,1\}^{k \times (m-k)}$.
		
		\item Each processor's output is simply the concatenation of its first $k$ random bits $x$ and $x^{T} M$.
	\end{itemize}
\end{theo}

That is, with $k$ random bits per processor as a seed, for every constant $c$, within $O(k)$ rounds we can turn them into $cn$ pseudo-random bits per processor which require $\Omega(k)$ rounds to be distinguished from random.

%\lnote{$\Omega(k)$ is confusing, just use $\Theta(k)$...}

%\begin{theo}
%\textbf{Note this is very informal right now}
%For every $k$ there exists a distribution $D_k$ such that the following hold:
%\begin{itemize}
%\item
%In the Broadcast Congested Clique, $D_k$ is sampleable within $O(k)$ rounds of congested clique, where each node must create only $O(k)$ bits to use as a seed.
%\item
%Within $k/10$ rounds, it is impossible for the Broadcast Congested Clique to distinguish between $D$ and the uniform distribution where each node receives $n$ random bits, except with probability exponentially low in $k$.
%\end{itemize}
%\end{theo}

So, for example, within $O(\log^2(n))$ round (In the $\BCAST(\log n)$ model, $O(\log n)$ rounds would suffice), one can construct a pseudo-random generator which is indistinguishable from random for any $\log^2 n$-round algorithm, except with low probability (the seed size for each processor would be $O(\log^2 n)$, while the size of the pseudo-random string is $\Theta(n)$).

The PRG is very simple to describe (we describe it here with seed size $2k$): first each processor shares $k$ random bits to create $k$ public random elements from $\mathbb{F}_2^n$ (i.e., $k$ random $n$-bit vectors). Then, each processor uses its remaining $k$ random bits to pick a random linear combination of those vectors (which requires $k$ bits to sample), and the result of this linear combination is the node's pseudo-random bits. So, essentially the pseudo-random generator is a distribution of low-rank matrices (which is very close to the uniform distribution matrices of rank up to $k$). Although the description and construction of the PRG are simple, and the seed size is tight up to a constant factor, the analysis is quite technical and involves new techniques. We remark that the pseudo-random generator has the additional nice property that it is computationally cheap; the only operations done by the processors is computing dot products of vectors over $\mathbb{F}_2$.

%We note that this doesn't let us derandomize, or even save random bits in a general way. We can show that for a k-round protocol, $O(k)$ bits per processor is always enough. However, the proof of this fact is via a counting argument. \lnote{add a proof? or a reference if exists?} One can view our result as a way to make this random bit saving transformation efficient.

\paragraph{Efficiently saving random bits}

It is possible to show that in the broadcast congested clique there is a randomized-deterministic separation (by reductions from two-player communication complexity for equality). That is, there are certain problems with faster randomized algorithms than the best possible deterministic algorithms. So, there is no hope for a general derandomization theorem. However, one can ask:  in general, what is the fewest number of random bits needed to efficiently solve problems in the Broadcast Congested Clique model? Using a technique of Newman \cite{newman1991private} from communication complexity, it is possible to show that for any protocol with output size $k$ per processor, $O(k)$ random bits per processor is enough (see Appendix \ref{newman}). The main downside of Newman's technique is that it is computationally inefficient: it holds in the case where all the processors have unbounded computational power, and is not a practical tool for saving random bits.

In this work, we ask the following question: In the Broadcast Congested Clique with \textit{computationally bounded (polynomial time)} processors, how many random bits are needed to perform general randomized computation? We can use our pseudo-random generator to show that every $k$-round algorithm where $k = \Omega(\log n)$ in which every processor uses up to $O(n)$ random bits can be transformed into an $O(k)$-round algorithm in which every processor uses up to $k$ random bits (that is, we can show that each processor needs to use at most $1$ random bit per round, while only increasing the run-time by a constant factor). Furthermore, this transformation is efficient. That is, if in the original algorithm all processors work in polynomial time, they also work in polynomial time in the new algorithm (in fact, there is only an \textit{additive} overhead of $O(kn)$ computation time for each processor. This overhead is the time required to compute the pseudo-random bits from the seed).

Stated in the setting of \BCAST($\log n$), where messages are of size $O(\log n)$ instead of $1$, we show that every $k$ round randomized algorithm using up to $n$ random bits per processor can be transformed into an $O(k)$ round algorithm using up to $O(k \log n)$ random bits per processor.

\paragraph{First \BCAST(1) Average Case Lower Bound and Hierarchy:} As a simple corollary of our PRG construction, we prove a $\BCAST(1)$ average case lower bound, which is the first average case lower bound proven in the model. Specifically, we show that when each processor receives a row of a sample from a certain close-to-uniform distribution of $n \times n$ matrices of rank $n-1$, this cannot be distinguished from each processor receiving $n$ uniformly random bits. We show that this implies that determining whether an input has rank $n$ or not is hard (it takes $\Omega(n)$ rounds), even when the input is chosen uniformly at random: 

\begin{theo}\label{theo:average-lowb}
	Let $n$ be a large enough integer and $\Ffullrank : \{0,1\}^{n \times n} \to \{0,1\}$ be the indicator function which indicates whether the given matrix has full rank. Suppose there are $n$ processors, and the $i$-th processor is given the $i$-th row of the input matrix. For all $n/20$-round \BCAST(1) protocols and all processors $i$ in it, $i$ cannot compute $F$ correctly with probability better than $0.99$ when the input is a uniformly chosen random matrix from $\{0,1\}^{n \times n}$.
\end{theo}

We also obtain an average-case time hierarchy theorem for the model:

\begin{theo}\label{theo:time-H}
	For any $\omega(\log n) \le k \le n$, there is a function $F$ such that a $k$-round $\BCAST(1)$ protocol can compute exactly, while any $k/20$-round $\BCAST(1)$ protocols cannot compute $F$ correctly with probability $0.99$ over the uniform distribution.
\end{theo}
%\onote{why $o(n)$? isn't $n$ ok?}
%\lnote{make sure it is understandable, Theorem 1.3 should be an average case lower bound but there is no distribution described in the theorem.}

%\lnote{maybe state Lemma~\ref{lm:dist} here...can be proved in the prelim}

%\subsubsection{An Abstract Framework}

\subsection{Toy Example: One Round Lower Bound for Planted Clique}\label{subsec:PlantedClique-Toy}
%!TEX root = main.tex
As a toy example to illustrate our proof framework, in the following we prove that the planted clique problem is hard for one-round $\BCAST(1)$ protocols when $k = o(n^{1/4})$. We begin with some notations.

\paragraph*{Notations.} 
Let $\distr_{m}$ denote the uniform distribution on $\{0,1\}^m$. For a function $f : \{0,1\}^{*} \to \{0,1\}^{*}$ and a distribution $\distrD$ on $\{0,1\}^*$, we use $f(\distrD)$ to denote the distribution of the output of $f$ when the input is drawn from $\distrD$. For two distributions $\distrD_1$ and $\distrD_2$, we use $\| \distrD_1 - \distrD_2\| = \frac12 \sum_{x \in \{0,1\}^*} |\distrD_1(x) - \distrD_2(x)|$ to denote their statistical distance (where $\distrD(x)$ is the probability that a sample from $\distrD$ equals $x$).

Let $\distrA^n_{\sf rand}$ be the distribution on $\{0,1\}^{n \times n}$ such that for a sample $A$ from $\distrA^n_{\sf rand}$, for all $i \ne j$, $A_{i,j}$ is an independent uniform random bit in $\{0,1\}$, and $A_{i,i}$ is always $0$ for all $i$. Let $C$ be a subset of $[n]$. We use $\distrA^{n}_{C}$ to denote the conditional distribution of $\distrA^n_{\sf rand}$ on the event that for all $i,j \in C$ and $i \ne j$, $A_{i,j} = 1$ (that is, $C$ is a clique). We also use $\distrA^{n}_{k}$ to be the mixed distribution of $\distrA^n_{C}$'s when $C$ is a uniform random subset of $[n]$ of size $k$. When the meaning is clear, we often drop the superscripts of the aforementioned distributions for simplicity.

So, to summarize, $\distrA^n_{\sf rand}$ is the uniform distribution over a random directed graph, $\distrA^{n}_{C}$ is the distribution where the vertices of $C$ are in a clique, and the rest of the edges are uniformly random, and $\distrA^{n}_{k}$ is the distribution where a random $k$ vertices are chosen to be a clique, and the rest of the edges are chosen uniformly at random.

By Yao's principle \cite{yaoprinciple}, we can assume all processors are deterministic as we are trying to prove a lower bound for distinguishing two input distributions. Processor $i$ can then be defined by a function $f_i : \{0,1\}^{n} \times \{0,1\}^{*} \to \{0,1\}$, such that $f_i(z,p)$ is the bit that processor $i$ outputs when it gets the input $z$ and transcript $p$.\footnote{In a zero-round protocol, the processor's $f_i$ does not take in an input $p$, since there is no transcript yet, just an input. However, in our proof, we assume that the processors broadcast their messages sequentially (that is, first the first processor speaks, then the second, and so on). In this stronger model, all but one of the processors do see a transcript before they broadcast their first bit.} We use $f_i^{|p}$ to denote the function $f_i(\cdot,p)$ for simplicity. If transcript $p$ is incompatible with processor $i$ having input $z$, then we set $f_i(z, p)$ arbitrarily.

Given a $\BCAST(1)$ protocol $\Pi$ and an input distribution $\distrD$, we use $\distrP(\Pi,\distrD)$ to denote the distribution of the transcripts of the protocol $\Pi$ running on a input drawn from $\distrD$ (that is, given a matrix $A$ which is drawn from the distribution $\distrD$, the processor $i$ gets the $i$-th row of $A$, and all processors act according to the protocol $\Pi$).

In this section we prove the following theorem.
\begin{theo}\label{theo:clique-lower-bound-one-round}
	Let $n$ be the number of processors and $k$ be an integer. For any one round \BCAST(1) protocol $\Pi$, we have
	\[
	\| \distrP(\Pi,\distrA_{\sf rand}) - \distrP(\Pi,\distrA_k) \| \le O\left(\frac{k^2}{\sqrt{n}}\right).
	\]
\end{theo}

That is, for any one-round protocol, the distribution of transcripts when run on a uniformly random input is statistically close to the distribution on an input with a planted clique. As a simple corollary, we immediately have:
\begin{cor}
	When $k = o(n^{1/4})$, no one-round $\BCAST(1)$ protocol $\Pi$ can distinguish between $\distrA_{\sf rand}$ and $\distrA_{k}$ with advantage $\Omega(1)$ (that is, any protocol which accepts on $\distrA_{\sf rand}$ with probability $p$, must accept on $\distrA_{k}$ with probability $p \pm o(1)$).
\end{cor}

That is, there is no way to solve the planted clique problem for cliques of size $o(n^{1/4})$ within one round of the Broadcast Congested Clique.

Let $\distrS_{k}^{T}$ be the uniform distribution on all size-$k$ subsets of $T$. To prove Theorem~\ref{theo:clique-lower-bound-one-round}, we need the following technical lemma, whose proof is deferred to the end of this section. In this lemma $f(x)$ represents the bit broadcasted by a processor, and $x$ represents the input to the processor. The lemma states that when picking a random clique $C$ of size $k$, then any function $f$ behaves similarly when the input is sampled from the uniform distribution vs. when the input is sampled from the uniform distribution with a clique planted on $C$. Basically, this means that for almost all possible cliques, $f$ does not substantially help distinguish between the clique existing, or the input being uniform.

\begin{lemma}\label{lm:close-subset-full-one-round}
	Let $n, k$ be integers such that $k \le n^{1/4}$, and $\distr_n^{C}$ be the uniform distribution on $\{ x : x \in \{0,1\}^n, x_i = 1 \text{ for all $i \in C$} \}$. For all function $f : \{0,1\}^n \to \{0,1\}$, we have
	\[
	\Ex_{C \sim \distrS_{k}^{[n]}} \left[ \left\| f(\distr_n) - f(\distr_n^{C}) \right\| \right] \le O\left( \frac{k}{\sqrt{n}} \right).
	\]
\end{lemma}

We also need the following lemma to bound the increase of the statistical distance when a processor speaks, The proof can be found in the preliminaries (Section \ref{sec:prelim}).

\begin{lemma}\label{lm:dist}
	Let $X$ and $Y$ be two sets, $\distrD$ and $\distrD'$ be two distributions on $X \times Y$. Let $\distrD_{|X}$ and $\distrD'_{|X}$ be the respective marginal distribution of $\distrD$ and $\distrD'$ on set $X$. For $a \in X$, we use $\distrD_{X=a}$ and $\distrD'_{X=a}$ to denote the respective conditional distribution of $\distrD$ and $\distrD'$ on $Y$ conditioning on $X = a$.\footnote{For simplicity, we let $\distrD_{X=a}$ be the uniform distribution on $Y$ if $\Pr_{(x,y) \sim \distrD}[x=a] = 0$.}
	We have
	\[
	\| \distrD - \distrD' \| \le \| \distrD_{|X} - \distrD'_{|X} \| + \Ex_{a \sim \distrD_{|X}} \left[ \| \distrD_{X = a} - \distrD'_{X = a} \| \right].
	\]
\end{lemma}

Now we are ready to prove Theorem~\ref{theo:clique-lower-bound-one-round}.

\begin{proofof}{Theorem~\ref{theo:clique-lower-bound-one-round}}
	
	Instead of viewing the algorithm as a single round algorithm, we will prove a slightly stronger lower bound. Consider the model where we have $n$ turns. On the $t^{th}$ turn, processor $t$ gets to send a single bit. This model is stronger than one round of the \BCAST(1) model, since it allows the later processors to condition their outputs on earlier the processors' messages. Hence, our lower bound implies a lower bound for the \BCAST(1) model as well.
	
	Let $\Prand^{(t)}$ and $\distrP_{C}^{(t)}$ be the distributions of the transcript of the first $t$ turns when the input is drawn from $\distrA_{\sf rand}$ or $\distrA_C$, respectively. Note that to prove the theorem, it suffices to show that $\Prand^{(n)}$ is close to most $\distrP_{C}^{(n)}$. For this purpose, we are going to prove the following inequality holds for any $t \le n$:
	
	\begin{equation}\label{eq:ex-distance-small-one-round}
	\Ex_{C \sim \distrS_{k}^{[n]}} \left[ \left\| \Prand^{(t)} - \distrP_{C}^{(t)} \right\| \right] \le t \cdot c_1 \cdot \frac{k^2}{n} \cdot \frac{1}{\sqrt{n}},
	\end{equation}
	
	where $c_1$ is a large enough universal constant. This inequality states that when picking a clique at random, the distribution of transcripts is similar to the distribution of transcripts when the input is chosen uniformly at random. It is easy to see that plugging in $t = n$, \eqref{eq:ex-distance-small-one-round} implies the theorem.		
	
	We prove the inequality above inductively. Clearly, \eqref{eq:ex-distance-small-one-round} holds when $t = 0$. So it suffices to show that when it holds for $t-1$, it also holds for $t$.
	
	Let $\distrD_t$ and $\distrD_t^{C}$ be the input distributions to processor $t$ in $\distrA_{\sf rand}$ and $\distrA_{C}$, respectively. For a fixed $C \subseteq [n]$, by Lemma~\ref{lm:dist}, we have:
	\begin{equation}\label{eq:derive-one-step-one-round}
	\left\| \Prand^{(t)} - \distrP_{C}^{(t)} \right\| \le \left\| \Prand^{(t-1)} - \distrP_{C}^{(t-1)} \right\| + \Ex_{p \sim \Prand^{(t-1)}} \left[ \left\|f_t^{|p}(\distrD_t) - f_t^{|p}(\distrD_t^C)  \right\| \right].
	\end{equation}
    
We think of $\| \Prand^{(t)} - \distrP_{C}^{(t)} \|$ as the amount of evidence the transcript is giving us about whether the distribution is uniform or if it has $C$ as a clique. (If this value were $0$, that would mean the transcript gives us no evidence. If the value were $1$, that would mean we have ``full" evidence and could distinguish between the two with no error given the transcript). So, with this interpretation, the inequality above is basically stating that the amount of evidence we have after round $t$ is the sum of the evidence from all rounds up to $t-1$, plus the extra evidence we get from the broadcast of the processor in round $t$.
	
	By definition, $\distrD_t$ is the uniform distribution on the set $\{ x : x \in \{0,1\}^n, x_t = 0 \}$. And $\distrD_t^{C} = \distrD_t$ if $t \notin C$, and is the uniform distribution on the set $\{ x : x \in \{0,1\}^n,x_t = 0,x_{j} = 1 \text{ for all $j \in C\setminus \{t\}$} \}$ otherwise.
	
	We care not about the probability of distinguishing a particular clique existing, but about whether any clique exists, so we take the expected value over all possible cliques of both sides of \eqref{eq:derive-one-step-one-round} gives:
	\begin{equation}\label{eq:need-to-prove-one-round}
	\Ex_{C \sim \distrS_{k}^{[n]}} \left[ \left\| \Prand^{(t)} - \distrP_{C}^{(t)} \right\| \right] \le \Ex_{C \sim \distrS_{k}^{[n]}} \left[ \left\| \Prand^{(t-1)} - \distrP_{C}^{(t-1)} \right\| \right] + \Ex_{p \sim \Prand^{(t-1)}} \Ex_{C \sim \distrS_{k}^{[n]}} \left[ \left\|f_t^{|p}(\distrD_t) - f_t^{|p}(\distrD_t^C)  \right\| \right].
	\end{equation}
	
	We can bound $\Ex_{C \sim \distrS_{k}^{[n]}} \left[ \left\| \Prand^{(t-1)} - \distrP_{C}^{(t-1)} \right\| \right]$ by the inductive hypothesis, so it suffices to bound \\$ \Ex_{p \sim \Prand^{(t-1)}} \Ex_{C \sim \distrS_{k}^{[n]}} \left[ \left\|f_t^{|p}(\distrD_t) - f_t^{|p}(\distrD_t^C)  \right\| \right]$. For $p \sim \Prand^{(t-1)}$, there are two cases:
	
	\begin{itemize}
		\item When $t \notin C$, which happens with probability $1 - \frac{k}{n}$, we have
		\[
		\left\|f_t^{|p}(\distrD_t) - f_t^{|p}(\distrD_t^C)  \right\| = 0,
		\]
		as $\distrD_t^C = \distrD_t$. That is, since $t$ is not in $C$, what $t$ says gives us no information about whether $C$ is a clique (whether there is a clique or not does not affect the input of $t$, and therefore does not affect their message). 
		
		\item When $t \in C$, which happens with probability $\frac{k}{n}$, by Lemma~\ref{lm:close-subset-full-one-round}, we have
		\[
		\Ex_{C' \sim \distrS_{k-1}^{[n] \setminus \{t\} } } \left[ \left\| f_t^{|p}(\distrD_t) - f_t^{|p}(\distrD_t^{C' \cup \{t\}} )  \right\| \right] \le O\left( \frac{k}{\sqrt{n}} \right).
		\]
	\end{itemize}
That is, when $t$ is in $C$, while $t$ might give information about whether $C$ is a clique, there are many cliques that may include $t$, and the inequality states that $t$ cannot give too much information about many of cliques (the expected amount of information revealed about a randomly chosen clique  of size $k$ containing $t$ is bounded by $O\left( \frac{k}{\sqrt{n}} \right)$).
	
	So, now we have bounded how much evidence the $t$-th processor reveals when broadcasting. We know that when the clique is chosen randomly, with probability $1 -k/n$ no information is revealed, and with probability $k/n$ at most $O(k/\sqrt{n})$ information is revealed in expectation. Combining these facts gives:
	\[
	\Ex_{p \sim \Prand^{(t-1)}} \Ex_{C \sim \distrS_{k}} \left[ \left\|f_t^{|p}(\distrD_t) - f_t^{|p}(\distrD_t^C)  \right\| \right]
\le \frac{k}{n} \cdot O \left( \frac{k}{\sqrt{n}} \right),
	\]
	which, plugging into \eqref{eq:need-to-prove-one-round} and using the inductive hypothesis proves inequality \eqref{eq:ex-distance-small-one-round} for $t$.
\end{proofof}

\subsubsection{Proof for Lemma~\ref{lm:close-subset-full-one-round}}

We need the following lemma first, whose proof is based on tools from information theory, and is deferred to the end of this subsection.\footnote{This lemma is standard and can be proved in various ways. We present a proof based on information theory because it can be easily generalized to a proof for Lemma~\ref{lm:close-subset-one}, which is used in Section~\ref{sec:clique-lower-bound}.}

\begin{lemma}\label{lm:close-subset-one-one-round}
	Let $n$ be an integer, and $\distr_n^{[i]}$ be the uniform distribution on $\{ x : x \in \{0,1\}^n, x_i = 1 \}$. For all function $f : \{0,1\}^n \to \{0,1\}$, we have
	\[
	\Ex_{ i \leftarrow [n] } \left[ \left\| f(\distr) - f(\distr^{[i]}) \right\| \right] \le O\left( \frac{1}{\sqrt{n}}  \right).
	\]
\end{lemma}

That is, if we consider a function $f$, suppose that on a uniform distribution, the probability it outputs $1$ is $p$. Then, if we pick a random index and set it to $1$, we still expect that if we randomly pick the rest of the coordinates, the output will be $1$ with probability approximately $p$.
%We decide not to include a proof of this lemma in the introduction, as it is quite technical. For a proof of the above lemma, see Lemma~\ref{lm:close-subset-one}, which is a generalization of the above lemma.

Now we are ready to prove Lemma~\ref{lm:close-subset-full-one-round} (restated below).

\begin{reminder}{Lemma~\ref{lm:close-subset-full-one-round}}
	Let $n, k$ be integers such that $k \le n^{1/4}$, and $\distr_n^{C}$ be the uniform distribution on $\{ x : x \in \{0,1\}^n, x_i = 1 \text{ for all $i \in C$} \}$. For all function $f : \{0,1\}^n \to \{0,1\}$, we have
	\[
	\Ex_{C \sim \distrS_{k}^{[n]}} \left[ \left\| f(\distr_n) - f(\distr_n^{C}) \right\| \right] \le O\left( \frac{k}{\sqrt{n}} \right).
	\]
\end{reminder}

The idea of the proof is that each bit we set to $1$, by Lemma \ref{lm:close-subset-one-one-round}, will change the expected output of $f$ by $O\left(\frac{1}{\sqrt{n}}\right)$. Hence, if we set $k$ of those bits to $1$, that will change the expected outcome by $O\left(\frac{1}{\sqrt{n}}\right)$ at most $k$ times, for a total of $O\left(\frac{k}{\sqrt{n}}\right)$. A formal proof is included below:

\begin{proofof}{Lemma~\ref{lm:close-subset-full-one-round}}
	
	\newcommand{\distrT}{\mathcal{T}}		
	Instead of choosing $C$ from $\distrS_k^{[n]}$, we choose an ordered $k$-tuple of $a = (a_1,a_2,\dotsc,a_k)$ of $k$ distinct elements in $[n]$ uniformly at random. Let the distribution be $\distrT_{k}^{[n]}$.
	
	We have
	\begin{align}
	\Ex_{C \sim \distrS_{k}^{[n]}} \left[ \left\| f(\distr_n) - f(\distr_n^{C}) \right\| \right] =& \Ex_{a \sim \distrT_{k}^{[n]}} \left[ \left\| f(\distr_n) - f(\distr_n^{ \{a_i\}_{i=1}^{k} }) \right\| \right] \notag \\
	\le& \sum_{\ell=1}^{k} \Ex_{a \sim \distrT_{\ell}^{[n]}} \left[ \left\| f(\distr_n^{ \{a_i\}_{i=1}^{\ell-1} }) - f(\distr_n^{ \{a_i\}_{i=1}^{\ell} }) \right\| \right] \notag \\
	\le& \sum_{\ell=0}^{k-1} \Ex_{a \sim \distrT_{\ell}^{[n]}} \Ex_{j \leftarrow [n] \setminus \{a_i\}_{i=1}^{\ell} } \left[ \left\| f(\distr_n^{ \{a_i\}_{i=1}^{\ell} }) - f(\distr_n^{ \{a_i\}_{i=1}^{\ell} \cup \{j\} }) \right\| \right]. \label{line:eq-term-one-round}
	\end{align}
	
	Now we are to bound the right side of \eqref{line:eq-term-one-round} for each $0 \le \ell \le k-1$ separately. Applying Lemma~\ref{lm:close-subset-one-one-round} on the restriction of $f$ such that all bits in $\{a_i\}$ are set to $1$, we have
	\[
	\Ex_{j \leftarrow [n] \setminus \{a_i\}_{i=1}^{\ell} } \left[ \left\| f(\distr_n^{ \{a_i\}_{i=1}^{\ell} }) - f(\distr_n^{ \{a_i\}_{i=1}^{\ell} \cup \{j\} }) \right\| \right] \le O\left(\frac{1}{\sqrt{n - \ell}}\right).
	\]
	
	Plugging the above in \eqref{line:eq-term-one-round} and noting that $k \le n^{1/4}$ completes the proof.
\end{proofof}

Now we prove Lemma~\ref{lm:close-subset-one-one-round}. The proof makes use several tools from information theory, see Section~\ref{sec:info-theory-prelim} for the details.

\begin{proofof}{Lemma~\ref{lm:close-subset-one-one-round}}
	
	Throughout the proof we will assume $X$ is a random variable drawn uniformly from $\{0,1\}^n$. For $i \in [n]$, let $X_i$ be the random variable of the $i$-th bit of $X$.
	
	We have
	\[
	I(X_i;f(X)) = H(X_i) - H(X_i | f(X)) = 1 - H(X_i | f(X)).
	\]
	
	And by the sub-additivity of conditional entropy, we have
	\[
	\sum_{i=1}^{n} H(X_i | f(X)) \ge H(X | f(X)) \ge n-1.
	\]
	Therefore,
	\[
	\sum_{i=1}^{n} I(X_i;f(X)) \le n - (n-1) \le 1.
	\]
	or equivalently,
	\[
	\Ex_{i \leftarrow [n]} I(X_i;f(X)) \le \frac{1}{n}.
	\]
	Note that by Fact~\ref{ft:KL-mutual},
	\[
	I(X_i ; f(X)) := \Ex_{x \sim X_i} D( f(X)_{X_i = x} || f(X) ).
	\]
	Taking expected values over $i$ of both sides and using $\Ex_{i \leftarrow [n]} I(X_i;f(X)) \le \frac{1}{n}$ gives
	\[		
	\Ex_{i \leftarrow [n]} I(X_i;f(X)) = \Ex_{i \leftarrow [n]} \Ex_{x \sim X_i} D( f(X)_{X_i = x} || f(X) ) \le \frac{1}{n}.
	\]
	By Pinsker's inequality (Lemma~\ref{lm:pinsker}) and the fact that $\sqrt{x}$ is a concave function, we have
	\[
	\Ex_{i \leftarrow [n]} \Ex_{x \sim X_i} \| f(X)_{X_i = x} - f(X) \| \le \sqrt{\frac{1}{n}},
	\]
	and
	\[
	\Ex_{i \leftarrow [n]} \frac{1}{2} \cdot \| f(X)_{X_i = 1} - f(X) \| \le \sqrt{\frac{1}{n}}.
	\]		
	
	Note that by definition, $f(X)_{X_i = 1}$ is distributed identically to $f(\distr^{[i]})$, which completes the proof.
\end{proofof}

%\lnote{check every notation is explained.}

%\lnote{make sure every lemmas can be understood without reading other stuffs}

\subsection{Related Work}

%\lnote{deal with this}

%\lnote{need to cite Pseudorandomness for network algorithms}

\paragraph{\BCAST(1) Congested Clique:} The specific distributed model we investigate is the Broadcast Congested Clique.
In this model, there are $n$ processors, and computation proceeds in rounds. In each round, each processor broadcasts a short message to all other processors. It has recently been studied in \cite{bcasttounicast, reconstruction, algebraic, erasures, power, morealgebraic, nelson2018optimal, approx, censor2016derandomizing, cyclelocality, MSFconnectivity, montealegre2016deterministic, jurdzinski2017brief}, among others. It has been used to study other areas in computer science such as streaming algorithms \cite{streaming} and mechanism design \cite{mechanism}.

\paragraph{Complexity Theoretic approaches in Distributed Computing:}
Recently, more complexity theoretic approaches and results have been made in the congested clique and distributed computation in general, for example in \cite{complexity, localcomplex, localcomplex2, hierarchy}.

\paragraph{Pseudo-randomness and Distributed Computing:}
In \cite{sharedcoins}, the authors construct a pseudo-random generator which creates additional shared random bits in a distributed system. Specifically, the authors work in a setting where every pair of processors can privately communicate with each other (whereas we work in the broadcast model), and some of the processors may be adversarially faulty. They show how to use few shared random bits, and unlimited private random bits to efficiently compute more shared random bits. In our setting, we are saving on private random bits flipped (when all processors are non-faulty, it is easy to turn a private random bit into a public random bit -- simply broadcast it).

In \cite{INW94}, the authors construct pseudo-randomness for a different distributed system, in which the network has a topology. Their main application is constructing pseudo-randomness that fools all low-space computation.

In \cite{distributedprg}, a different setting than ours is considered, in which the processors are computationally bounded, and a cryptographic pseudo-random function is being evaluated.

In \cite{meravprg}, a pseudo-random generator that fools DNFs is used to deterministically construct spanners in the congested clique. In that work, the pseudo-random generator is used for the specific problem considered, as opposed to being a pseudo-random generator which fools all algorithms in the model.

Pseudo-randomness in the context of complexity theory has been very widely studied. See Vadhan's survey \cite{vadhanpseudorand}.

In~\cite{GHK18}, the authors introduce general methods for derandomizing algorithms in the LOCAL model to obtain better deterministic algorithms. 

%\lnote{need to add some related works on planted clique}

\paragraph{Planted Clique:} The planted clique problem (or hidden clique problem) was introduced in~\cite{Jer92} and~\cite{Kuc95}. 
%As a natural average case problem, it is connected to many other problems across different areas of computer science, such as, ...,...,...,...
The best known classical algorithm~\cite{FK00,DGP14} can find the hidden clique when its size is $k = \Omega(\sqrt{n})$ in near linear-time. For $k \ll \sqrt{n}$, the na\"{\i}ve algorithm (looking for a clique of size $10\log n$ with brute force, and then extending that clique to the whole clique) can solve it in $n^{O(\log n)}$ time, and the problem is conjectured to be not solvable in polynomial time. However, since it is an average-case problem, it is unlikely that the hardness of this problem can be derived from standard complexity assumptions such as $\textsf{P} \ne \textsf{NP}$~\cite{FF93,BT06}. Therefore, much work has been put into trying to show limitations for certain classes of algorithm on this problem~\cite{FK03,MPW15,DM15,HOP+2018,BHKKMP16}, or showing tight hardness for closely-related worst-case problems under standard assumptions~\cite{BKRW17}.

\paragraph{Distributed Clique lower bounds}
There is some literature on lower bounds for finding cliques in the congested clique model \cite{power}, as well as the standard CONGEST model \cite{cliqueCONGEST}. These lower bounds hold in the worst case, and have no direct implications about the hardness of the planted clique problem.

\begin{comment}
\subsection{Random Notes}

note that in ucast, our prg is not pseudo-random

note why this doesn't help speed up algorithms

we are essentially optimal (maybe up to logn factors)

for many distributed (and bcast in particular) algos (give examples?) the most expensive operation seems to be creating the random bits. 

Note that the single bit prg is just distinguishing rank k vs k+1, or something like that.

What is pseudo-randomness: with high probability indistinguishable from actually random
\end{comment}

\section*{Organization of the Paper}
In Section~\ref{sec:prelim}, we introduce the needed preliminaries for this paper. Section~\ref{sec:abs-framework} we present an abstract framework for our approach. In Section~\ref{sec:clique-lower-bound} we prove the lower bound for planted clique in $\BCAST(1)$. In Section~\ref{sec:PRG-overview} we give an overview of the proofs for the PRG construction for $\BCAST(1)$, starting with a one-round toy example. Then in Section~\ref{sec:one-secret}, we show how to create a single pseudo-random bit for each processor, which also implies our average case lower bound for $\BCAST(1)$. Next, in Section~\ref{sec:completePRG} we show how to create many pseudo-random bits. Finally, in Section~\ref{sec:lowerbound} we show our pseudo-random generator's parameters are optimal. 

\section{Preliminaries}\label{sec:prelim}

\newcommand{\distrQ}{\mathcal{Q}}

%\subsection{Notations}\label{sec:notations}

\subsection{Notations}

Here we summarize some standard notations which are used in this paper. 

For integers $n$ and $m$, we use $\distr_{m}$ to denote the uniform distribution on $\{0,1\}^m$, and $\distr_{n \times m}$ to denote the uniform distribution on $\{0,1\}^{n \times m}$.

Let $X,Y$ be two sets. For a function $f : X \to Y$ and a distribution $\distrD$ on $X$, we use $f(\distrD)$ to denote the distribution of the output of $f$ when the input is drawn from $\distrD$. For two distributions $\distrD_1$ and $\distrD_2$ on a set $X$, we use $\| \distrD_1 - \distrD_2\| = \frac12 \sum_{x \in X} |\distrD_1(x) - \distrD_2(x)|$ to denote their statistical distance (where $\distrD(x)$ is the probability that a sample from $\distrD$ equals $x$).

\subsection{Analysis of Boolean Functions}
\label{sec:analysis-of-boolfunc}

Our proofs make use of some well-known facts from analysis of Boolean functions\footnote{Some nice references can be found in~\cite{de2008brief,o2014analysis}}.

For any function $f : \{0,1\}^{n} \to \mathbb{R}$, its \emph{Fourier coefficient} at a set $S$ is defined as
\[
\WH{f}(S) := \Ex_{x \sim \distr_{n} }  \left[f(x) \cdot (-1)^{\sum_{i \in S} x_i} \right].
\]

Parseval’s Identity states
\[
\Ex_{x \sim \distr_n} \left[f(x)^2\right] = \sum_{S \subseteq [n]} \WH{f}(S)^2.
\]

\subsection{Probability Theory}

The following lemma is standard. We provide a proof here for completeness.

\begin{reminder}{Lemma~\ref{lm:dist}}
	Let $X$ and $Y$ be two sets, and $\distrD$ and $\distrD'$ be two distributions on $X \times Y$. Let $\distrD_{|X}$ and $\distrD'_{|X}$ be the respective marginal distribution of $\distrD$ and $\distrD'$ on set $X$. For $a \in X$, we use $\distrD_{X=a}$ and $\distrD'_{X=a}$ to denote the respective conditional distribution of $\distrD$ and $\distrD'$ on $Y$ conditioning on $X = a$.\footnote{For simplicity, we let $\distrD_{X=a}$ be the uniform distribution on $Y$ if $\Pr_{(x,y) \sim \distrD}[x=a] = 0$.}
	We have
	\[
	\| \distrD - \distrD' \| \le \| \distrD_{|X} - \distrD'_{|X} \| + \Ex_{a \sim \distrD_{|X}} \left[ \| \distrD_{X = a} - \distrD'_{X = a} \| \right].
	\]
\end{reminder}
\newcommand{\aux}{\textsf{aux}}
\begin{proof}
	We first define an auxiliary distribution $\distrD_\aux$ as follows: for $(a,b) \in X \times Y$, if $\distrD'_{|X}(a) > 0$, (we use $\distrD'_{|X}(a)$ to denote the probability that a sample from $\distrD'_{|X}$ equal $a$).
	
	\[
	\distrD^\aux(a,b) := \distrD'(a,b) \cdot \frac{\distrD_{|X}(a)}{\distrD'_{|X}(a)}.
	\]
	
	Otherwise, we set $\distrD^\aux(a,b) := \frac{1}{|Y|} \cdot \distrD_{|X}(a)$. It is easy to verify that $\distrD^\aux_{|X} = \distrD_{|X}$ and for all $a$ $\distrD^\aux_{X = a} = \distrD'_{X = a}$, and therefore it is a distribution.
	
	Now, it is easy to see that
	\[
	\| \distrD^\aux - \distrD' \| =  \| \distrD_{|X} - \distrD'_{|X} \|.
	\]
	
	Moreover, we have
	\[
	\| \distrD^\aux - \distrD \| = \Ex_{a \sim \distrD_{|X}} \left[ \| \distrD_{X = a} - \distrD'_{X = a} \| \right].
	\]
	
	Putting everything together, we have
	\begin{align*}
	\| \distrD - \distrD' \| &\le 	\| \distrD^\aux - \distrD' \| + \| \distrD^\aux - \distrD \|.\\
	&\le \| \distrD_{|X} - \distrD'_{|X} \| + \Ex_{a \sim \distrD_{|X}} \left[ \| \distrD_{X = a} - \distrD'_{X = a} \| \right].
	\end{align*}
	
\end{proof}

%\lnote{explain why we need this lemma}

\subsection{Information Theory}\label{sec:info-theory-prelim}

In this paper we need some definitions and facts from information theory. For an excellent introduction to information theory, one is referred to the textbook by Cover and Thomas~\cite{CT06}. We consider discrete random variables in this paper.

%Defn's of Entropy, Conditional Entropy, Mutual Information and KL Divergence.

Let $X,Y$ be random variables in the same probability space $\Omega$. The entropy of $X$, denoted by $H(X)$, is defined as $H(X) := \Pr_{a \sim X} \log \frac{1}{\Pr[X = a]} $. The conditional entropy of $X$ given $Y$, denoted as $H(X|Y)$, is defined as $H(X|Y) := \Ex_{y \sim Y} H(X | Y = y)$.

The mutual information between $X$ and $Y$, denoted by $I(X;Y)$, is defined as $I(X;Y) := H(X) - H(X|Y) = H(Y) - H(Y|X) = H(X) + H(Y) - H(X,Y)$. 

For two distributions $\distrP$ and $\distrQ$ on the same set $S$, their Kullback-Leibler (KL) divergence is defined as
\[
D(\distrP || \distrQ) := \sum_{s \in S} \distrP(s) \cdot \log \frac{\distrP(s)}{\distrQ(s)}.
\]

KL divergence is related to mutual information in the following way.

\begin{ft}\label{ft:KL-mutual}
	\[
	I(X; Y) := \Ex_{x \sim X} D( Y | X =x || Y ),
	\]
\end{ft}

\begin{lemma}[Pinsker's Inequality]\label{lm:pinsker}
	For two distributions $\distrD_1$ and $\distrD_2$, we have
	\[
	\| \distrD_1 - \distrD_2 \| \le \sqrt{ \frac{1}{2} \cdot D(\distrD_1 || \distrD_2)}.
	\]
\end{lemma}

\newcommand{\Ber}{\textsf{Ber}}

For a real $p \in [0,1]$, we use $\Ber(p)$ to denote the binary Bernoulli random variable with expectation $p$. We also use $H(p)$ to denote the $H(\Ber(p))$. We have the following fact.

\begin{ft}\label{ft:H-p-property}
	If $H(p) \ge 0.9$, we have $p \in [0.3,0.7]$, and
	\[
	\frac{1-H(p)}{(p-1/2)^2} \in [2,3].
	\]
\end{ft}
\newcommand{\distrpseduo}{\distrA_{\sf pseudo}}
\newcommand{\distrrand}{\distrA_{\sf rand}}
\newcommand{\progressFunc}{\mathcal{L}_{\sf progress}}
\newcommand{\realDist}{\mathcal{L}_{\sf real\text{-}dist}}
\newcommand{\Ppseduo}{\distrP_{\sf pseudo}}

\section{Abstract Framework}\label{sec:abs-framework}
In this section, we present an abstraction of our framework. Understanding this section is not necessary to understand the rest of the sections of the paper. It is included to make it easier to understand the structure of the proof without having to dig through the problem-specific technical parts. %It may be useful for readers who wish to apply our technique to another problem, but do not want to understand all the technical details of the planted clique or the pseudo-random generator lower bounds.

In the following we exhibit an abstract framework for showing a certain input distribution $\distrpseduo$ is indistinguishable from the uniform random input distribution $\distrrand$ by a low round $\BCAST(1)$ protocol\footnote{It is not imperative that one of the distributions is uniform. We decide to present the framework with one of the distributions as uniform for the sake of simplicity, and since our two main applications of the framework in this paper involve distinguishing distributions from uniform.}.

For the simplicity of discussion. We assume each of the $n$ processors gets $n$ bits as its input. We also use a matrix $A \in \{0,1\}^{n \times n}$ to denote their inputs collectively, where the $i$-th player gets the $i$-th row of $A$. Then $\distrrand$ is simply the uniform distribution over $\{0,1\}^{n \times n}$.

%\lnote{Note that this is not required to understand the proof, but good to know if you want to adopt it for more applications.}

\paragraph*{Notations.} 
We first recall and introduce some notations. Let $\distr_{m}$ denote the uniform distribution on $\{0,1\}^m$. For a function $f : \{0,1\}^{*} \to \{0,1\}^{*}$ and a distribution $\distrD$ on $\{0,1\}^*$, we use $f(\distrD)$ to denote the distribution of the output of $f$ when the input is drawn from $\distrD$. For two distributions $\distrD_1$ and $\distrD_2$, we use $\| \distrD_1 - \distrD_2\| = \frac12 \sum_{x \in \{0,1\}^*} |\distrD_1(x) - \distrD_2(x)|$ to denote their statistical distance (where $\distrD(x)$ is the probability that a sample from $\distrD$ equals $x$).

By Yao's principle \cite{yaoprinciple}, we can assume all processors are deterministic as we are trying to prove a lower bound for distinguishing two input distributions. Processor $i$ can then be defined by a function $f_i : \{0,1\}^{n} \times \{0,1\}^{*} \to \{0,1\}$, such that $f_i(z,p)$ is the bit that player $i$ outputs when it gets the input $z$ and transcript $p$. We use $f_i^{|p}$ to denote the function $f_i(\cdot,p)$ for simplicity. If transcript $p$ is incompatible with player $i$ having input $z$, then we set $f_i(z, p)$ arbitrarily.

Given a $\BCAST(1)$ protocol $\Pi$ and an input distribution $\distrD$, we use $\distrP(\Pi,\distrD)$ to denote the distribution of the transcripts of the protocol $\Pi$ running on a input drawn from $\distrD$. We also use $\distrP^{(t)}(\Pi,\distrD)$ to denote the distribution of the same transcript in first $t$ turns.

\subsubsection*{A Relaxation}

Instead of viewing the algorithm as a single round algorithm, we will prove a slightly stronger lower bound. Consider the model where we have $j \cdot n$ turns instead of $j$ rounds. On the $t^{th}$ turn, processor $(t-1) \bmod{n} + 1$ gets to send a single bit. This model is stronger than $j$ rounds of the \BCAST(1) model, since it allows the later processors to condition their outputs on earlier processors' messages. Hence, lower bounds for this relaxed model imply lower bounds for the \BCAST(1) model as well.

\subsubsection*{Decomposition into Row-Independent Distributions}

\newcommand{\IndexSet}{\mathcal{I}}

We first write $\distrpseduo$ as an average of many row-independent distributions. Let $\IndexSet$ be an index set, and $\{ \distrA_{I} \}_{I \in \IndexSet}$ be a family of distributions, we need the following two properties:

\begin{itemize}
	\item $\distrpseduo = \frac{1}{|\IndexSet|} \sum_{I \in \IndexSet} \distrA_{I}$. That is, $\distrpseduo$ can be written as an average of all distributions in $\{ \distrA_{I} \}_{I \in \IndexSet}$.
	\item For each $I \in \IndexSet$, $\distrA_I =  \bigoplus_{i=1}^{n} \distrA_I^{[i]} $, where $\oplus$ means concatenation and all $\distrA_I^{[i]}$'s are independent. Equivalently, rows in $\distrA_{I}$ are \emph{independent}. (Each row is a single node's input).
\end{itemize}

\subsubsection*{Progress Function}
We first fix a $\BCAST(1)$ protocol $\Pi$. For simplicity, we define $\Prand^{(t)} = \distrP^{(t)}(\Pi,\distrrand)$, $\Ppseduo^{(t)} = \distrP^{(t)}(\Pi,\distrpseduo)$ and $\distrP^{(t)}_{I} = \distrP^{(t)}(\Pi,\distrA_{I})$ for $I \in \IndexSet$.

 Ideally, we would like to bound
\[
\realDist^{(t)} :=  \left\| \Ppseduo^{(t)} - \Prand^{(t)} \right\|
\]
round by round. But as discussed in the introduction, the above is very hard to work with, so we try to bound the following progress function instead:

\[
\progressFunc^{(t)} := \Ex_{I \leftarrow \IndexSet} \left[ \left\| \distrP^{(t)}_I - \Prand^{(t)} \right\| \right].
\]

It is not hard to see that $\realDist^{(t)} \le \progressFunc^{(t)}$: we know that $\progressFunc^{(t)} = \frac{1}{|\IndexSet|} \sum_{I \in \IndexSet} \left\| \distrP^{(t)}_I - \Prand^{(t)} \right\|$, which by a triangle inequality is greater than or equal to $  \left\| \frac{1}{|\IndexSet|} \sum_{I \in \IndexSet} \left[ \distrP^{(t)}_I - \Prand^{(t)}\right] \right\| = \left\| \Ppseduo^{(t)} - \Prand^{(t)} \right\|$ so showing an upper bound on $\progressFunc^{(t)}$ is sufficient for upper bounding $\realDist^{(t)}$. 

\subsubsection*{Upper Bounding the Progress Made in Turn $t$}

Now suppose we are at the $t$-th turn. Let $j$ be the current round number, and $i$ be the broadcasting processor of this turn. By Lemma~\ref{lm:dist}, for all $I \in \IndexSet$, we have
\begin{equation}\label{eq:simple-eq}
\left\| \distrP^{(t)}_I - \Prand^{(t)} \right\| \le \left\| \distrP^{(t-1)}_I - \Prand^{(t-1)} \right\| + \Ex_{p \sim \Prand^{(t-1)}} \left[ \left\|f_i^{|p}(\distrD_i | p) - f_i^{|p}(\distrD_i^I | p)  \right\| \right].
\end{equation}

In above, $\distrD_i | p$ and $\distrD_i^{I} | p$ are the input distributions to player $i$ conditioning on seeing the transcript $p$ of the previous $t-1$ rounds, in $\distrrand$ and $\distrA_{I}$ respectively. 

First, since $\distrD_i$ is just $\distr_{n}$, we can see $\distrD_i | p$ is simply the uniform distribution on the set of inputs which is consistent with the transcript $p$. Formally, let $t_1,t_2,\dotsc,t_{j-1}$ be the indices of all previous $j-1$ turns with processor $i$ broadcasting, before the current $t$-th turn. For $x \in \{0,1\}^{n}$, we say that $x$ is consistent with transcript $p$, if for all $\ell \in [j-1]$, we have
\[
f_i^{|p^{(t_\ell - 1)}}(x) = p_{t_\ell},
\]
where $p^{(t_\ell - 1)}$ denotes the first $t_\ell - 1$ bits of $p$. That is, simulating $f_i$ with respect to $p$ on $x$ gives the same outputs in $p$.

%\onote{$a$ is never defined above}

Let $D^{(t-1)}_p$ denote the set of inputs to $f_i$ which are consistent with the transcript $p$. Then we can see $\distrD_i | p$ is the uniform distribution on $D^{(t-1)}_p$. Similarly, since $\distrA_I$ is row-independent, $\distrD_i^I | p$ is just $\distrA_I^{[i]}$ conditioning on $D^{(t-1)}_p$. We denote this as $\distrA_I^{[i]} | D^{(t-1)}_p$.

Plugging in~\eqref{eq:simple-eq}, and taking an expectation for all $I \in \IndexSet$, we have

\begin{equation}\label{eq:bound-next-step-abs}
\Ex_{I \leftarrow \IndexSet}\left\| \distrP^{(t)}_I - \Prand^{(t)} \right\| \le \Ex_{I \leftarrow \IndexSet} \left\| \distrP^{(t-1)}_I - \Prand^{(t-1)} \right\| + \Ex_{p \sim \Prand^{(t-1)}} \Ex_{I \leftarrow \IndexSet} \left[ \left\|f_i^{|p}(\distr_{D^{(t-1)}_p}) - f_i^{|p}(\distrA_I^{[i]} | D^{(t-1)}_p)  \right\| \right].
\end{equation}

A key observation here is that $\distrD^{(t-1)}_p$ is usually a large set over $p \sim \Prand^{(t-1)}$. The proof of the following claim is essentially the same as the proof for Claim~\ref{claim:large-D-p-clique} in Section \ref{sec:clique-lower-bound}, so we omit it here. %\onote{Do you mean "the claim below"? Also, if it's short and basically a copy-paste, I would just include the proof here. Otherwise this is fine, even though it's a little bit bad that we are referring to something later in the paper...}

\begin{claim}
	For all $\varepsilon > 0$,
	\[
	\Pr_{p \sim \Prand^{(t-1)}} \left[ D^{(t-1)}_p \ge 2^{n - j} \cdot \epsilon \right] \ge 1 - \epsilon.
	\]
\end{claim}

Therefore, in order to bound the second term of the right hand side of \eqref{eq:bound-next-step-abs}, we can assume $|D^{(t-1)}_p| \ge 2^{n - \Theta(\beta)}$, where $\beta$ is roughly the round lower bound we wish to prove.

\subsubsection*{Statistical Inequality Task}

Now we are finally able to specify the statistical inequality task we need to prove. We want to show that for almost all $p \sim \Prand^{(t-1)}$, their contribution to the second term of the right side of \eqref{eq:bound-next-step-abs},
\[
\Ex_{I \leftarrow \IndexSet} \left[ \left\|f_i^{|p}(\distr_{D^{(t-1)}_p}) - f_i^{|p}(\distrA_I^{[i]} | D^{(t-1)}_p)  \right\| \right],
\]  
is small. We obviously have no control over the set $D^{(t-1)}_p$ except for that it is large, so we want the following type of statistically inequality.

%\onote{I think the term "lemma" is bad here, since we aren't actually proving that this is true -- Instead maybe let's call it a "Required Lemma Format"}
\bigskip
\noindent {\bf Required Lemma Format. }{\em Let $D \subseteq \{0,1\}^{n}$ with $|D| \ge 2^{n-\beta}$, $\distr_{D}$ be the uniform distribution on $D$. For all function $f : D \to \{0,1\}$ and $i \in [n]$, we have
\[
	\Ex_{I \leftarrow \IndexSet} [ \| f(\distr_D) - f(\distrA_{I}^{[i]} | D) \| ] \le \epsilon(n,\beta).
\]
}

In above $\epsilon(n,\beta)$ is some error function which is increasing in $\beta$.

It will be helpful to observe that Lemma~\ref{lm:close-subset-full-one-round}, Lemma~\ref{lm:close-subset-full}, Lemma~\ref{lm:fourier}, Lemma~\ref{lm:fourier2}, and Lemma~\ref{lm:fourier4} are all instantiations of the above required lemma (for proving the one-round lower bound we can simply assume $D = \{0,1\}^n$).

Once we have the required lemma, then by a simple induction, we have
\[
\progressFunc^{(j \cdot n)} \le \sum_{\ell=1}^{j} \epsilon(n,\Theta(\ell))) \le (j \cdot n) \cdot \epsilon(n,\Theta(j)).
\]

From which we can deduce the needed lower bound, if $(j \cdot n) \cdot \epsilon(n,\Theta(j)) \ll 1$.
\section{Lower Bound for Planted Clique}
\label{sec:clique-lower-bound}

In this section we prove that the planted clique problem is hard for $\BCAST(1)$ when $k = n^{1/4 - \epsilon}$. We encourage the reader to read Section \ref{subsec:PlantedClique-Toy} before this section. That subsection contains a one-round lower bound for the problem, which involves a similar yet much less technical proof.

\paragraph*{Notations.} 
We first recall some notations. Let $\distrA^n_{\sf rand}$ be the distribution on $\{0,1\}^{n \times n}$ such that for a sample $A$ from $\distrA^n_{\sf rand}$, for all $i \ne j$, $A_{i,j}$ is an independent uniform random bit in $\{0,1\}$, and $A_{i,i}$ is always $0$ for all $i$. Let $C$ be a subset of $[n]$. We use $\distrA^{n}_{C}$ to denote the conditional distribution of $\distrA^n_{\sf rand}$ on the event that for all $i,j \in C$ and $i \ne j$, $A_{i,j} = 1$ (that is, $C$ is a clique). We also use $\distrA^{n}_{k}$ to denote the mixed distribution of $\distrA^n_{C}$'s when $C$ is a uniformly chosen random subset of $[n]$ of size $k$. 

For a distribution $\distrA$, we use $\distrA^{[i]}$ to denote it's marginal distribution on the $i$-th row. Note that $\distrA^n_{\sf rand}$ and $\distrA^n_C$ have independent rows\footnote{Fixing a clique $C$, all entries of the distribution $\distrA^n_C$ are independent: each edge outsize of $C$ is an independent coin flip with probability 1/2, and each edge in $C$ is an independent coin flip with probability 1. In particular, note that every two edges in the clique are independent, since they are both 1 with probability 1, and therefore the mutual information between the two entries is 0.}.

When the meaning is clear, we often drop the superscripts of the above distributions for simplicity.

Given a $\BCAST(1)$ protocol $\Pi$ and an input distribution $\distrD$, we use $\distrP(\Pi,\distrD)$ to denote the distribution of the transcripts of the protocol $\Pi$ running on an input drawn from $\distrD$ (that is, given a matrix $A$ which is drawn from the distribution $\distrD$, the processor $i$ gets the $i$-th row of $A$, and all processors act according to the protocol $\Pi$).

In this section we prove the following theorem:
\begin{theo}\label{theo:clique-lower-bound}
	Let $n$ be the number of processors. For any $j$-round \BCAST(1) protocol $\Pi$, we have
	\[
	\| \distrP(\Pi,\distrA_{\sf rand}) - \distrP(\Pi,\distrA_k) \| \le O\left(j \cdot k^2 \cdot \sqrt{\frac{j + \log n}{n}}\right).
	\]
\end{theo}

As a simple corollary, we immediately have:
\begin{cor}[$\BCAST(1)$ Lower Bound for Planted Clique]\label{cor:mainclique}
	For any constant $\epsilon > 0$, if $k = n^{1/4 - \epsilon}$ then no $n^{o(1)}$ round $\BCAST(1)$ protocol $\Pi$ can distinguish between $\distrA_{\sf rand}$ and $\distrA_{k}$ with advantage $\Omega(1)$.
\end{cor}

Let $\distrS_{k}^{T}$ be the uniform distribution on all size-$k$ subsets of $T$. To prove Theorem~\ref{theo:clique-lower-bound}, we need the following technical lemma, whose proof is deferred to the end of this section.

\begin{lemma}\label{lm:close-subset-full}
	Let $n, t, k$ be integers such that $t,k \le n^{1/4}$ and $t \ge 10\log n$, $D$ be a subset of $\{0,1\}^n$ with $|D| \ge 2^{n-t}$, $\distr_{D}$ be the uniform distribution on $D$, and $\distr_{D}^{C}$ be the uniform distribution on $\{ x : x \in D, x_i = 1 \text{ for all $i \in C$} \}$. For all functions $f : D \to \{0,1\}$, we have
	\[
	\Ex_{C \sim \distrS_{k}^{[n]}} [ \| f(\distr_D) - f(\distr_D^{C}) \| ] \le O\left(k \cdot \sqrt{\frac{t}{n}}\right).
	\]
    (If $\distr_D^{C}$ is empty, we define $\| f(\distr_D) - f(\distr_D^{C}) \| = 1$).
\end{lemma}

Intuitively speaking, the $D$ in the lemma above corresponds to the set of inputs to a certain node which are consistent with the current transcript. Each time the node broadcasts a bit, the size of $D$ is expected to reduce by at most a constant factor, so after $r$ rounds one would expect $D$ to be larger than $2^{n-\Theta(r)}$.

Now we are ready to prove Theorem~\ref{theo:clique-lower-bound}.

\begin{proofof}{Theorem~\ref{theo:clique-lower-bound}}
	
	Instead of viewing the algorithm as a standard $j$ round algorithm, we will prove a slightly stronger lower bound. Consider the model where during each round we have $n$ turns. On the $t^{th}$ turn, processor $(t-1) \bmod{n} + 1$ gets to send a single bit. So, essentially, instead of all processors broadcasting their bit at the same time, they take turns. This model is stronger than one round of the \BCAST(1) model, since it allows the later processors to condition their outputs on earlier the processors' messages. Hence, our lower bound implies a lower bound for the \BCAST(1) model as well.
	
	Let $\Prand^{(t)}$ and $\distrP_{C}^{(t)}$ be the distributions of the transcript of the first $t$ turns when the input is drawn from $\distrA_{\sf rand}$ or $\distrA_C$, respectively. Note that to prove the theorem, it suffices to show that the distribution $\Prand^{(j \cdot n)}$ is close to $\distrP_{C}^{(j \cdot n)}$ for most choices of $C$. For this purpose, we are going to prove the following inequality holds for any $t \le j \cdot n$:
	
	\begin{equation}\label{eq:ex-distance-small}
	\Ex_{C \sim \distrS_{k}^{[n]}} \left[ \left\| \Prand^{(t)} - \distrP_{C}^{(t)} \right\| \right] \le t \cdot \left( 1/n^2 +  c_1 \cdot \frac{k^2}{n} \cdot \sqrt{\frac{j+\log n}{n}}\right),
	\end{equation}
	
	where $c_1$ is a large enough universal constant. It is easy to see that plugging in $t = j \cdot n$, \eqref{eq:ex-distance-small} implies the theorem.		
	
	To prove \eqref{eq:ex-distance-small}, we induct on $t$. Clearly, \eqref{eq:ex-distance-small} holds when $t = 0$. So it suffices to show that when it holds for $t-1$, it also holds for $t$. Let $i$ be the processor who is broadcasting at the $t$-th turn.
	
	For a fixed $C \subseteq [n]$, by Lemma~\ref{lm:dist}, we have:
	\begin{equation}\label{eq:derive-one-step}
	\left\| \Prand^{(t)} - \distrP_{C}^{(t)} \right\| \le \left\| \Prand^{(t-1)} - \distrP_{C}^{(t-1)} \right\| + \Ex_{p \sim \Prand^{(t-1)}} \left[ \left\|f_i^{|p}(\distrD_i | p) - f_i^{|p}(\distrD_i^C | p)  \right\| \right].
	\end{equation}
	
	In above, $\distrD_i | p$ and $\distrD_i^{C} | p$ are the input distributions to player $i$ conditioning on seeing the transcript $p$ of the previous $t-1$ rounds. Let $D^{(t-1)}_p$ denote the set of inputs from $\{ x : x \in \{0,1\}^n, x_i = 0 \}$ to $f_i$ which are consistent with the transcript $p$. We can see $\distrD_i | p$ is the uniform distribution on $D^{(t-1)}_p$, while $\distrD_i^C | p$ is the same as $\distrD_i |p$ when $i \not\in C$, and is the uniform distribution on $ \{ x : x \in D^{(t-1)}_p, x_j = 1 \text{ for all $j \in C \setminus \{i\} $} \}$.
	
	Taking the expected value over all cliques of both sides of \eqref{eq:derive-one-step} gives
	\begin{equation}\label{eq:need-to-prove}
	\Ex_{C \sim \distrS_{k}^{[n]}} \left[ \left\| \Prand^{(t)} - \distrP_{C}^{(t)} \right\| \right] \le \Ex_{C \sim \distrS_{k}^{[n]}} \left[ \left\| \Prand^{(t-1)} - \distrP_{C}^{(t-1)} \right\| \right] + \Ex_{p \sim \Prand^{(t-1)}} \Ex_{C \sim \distrS_{k}^{[n]}} \left[ \left\|f_i^{|p}(\distrD_i | p) - f_i^{|p}(\distrD_i^C | p)  \right\| \right].
	\end{equation}
	
	So, to prove \eqref{eq:ex-distance-small}, since we can bound $\Ex_{C \sim \distrS_{k}^{[n]}} \left[ \left\| \Prand^{(t-1)} - \distrP_{C}^{(t-1)} \right\| \right]$ by the inductive hypothesis,
 it suffices to bound $ \Ex_{p \sim \Prand^{(t-1)}} \Ex_{C \sim \distrS_{k}^{[n]}} \left[ \left\|f_i^{|p}(\distrD_i | p) - f_i^{|p}(\distrD_i^C | p)  \right\| \right]$. We first show that for most $p \sim \Prand^{(t-1)}$, $D^{(t-1)}_p$ is large (that is, after $t-1$ turns, we expect the set of inputs consistent with the transcript to be large). The proof for the following claim is deferred to the end of the whole proof.
	
	\begin{claim}\label{claim:large-D-p-clique}
		For $t \le j \cdot n \le \frac{k\cdot n}{10}$, with probability $1 - 1 / n^2$ over $p \sim \Prand^{(t-1)}$, we have $|D^{(t-1)}_p| \ge  2^{n - j} / n^3$.
	\end{claim}

	Now, given a $p$ with $|D^{(t-1)}_p| \ge  2^{n - j} / n^3 = 2^{n-j - 3 \log n}$, we want to bound $\Ex_{C \sim \distrS_{k}^{[n]}} \left[ \left\|f_i^{|p}(\distrD_i | p) - f_i^{|p}(\distrD_i^C | p)  \right\| \right]$. There are two cases:
	
	\begin{itemize}
		\item When $i \notin C$, which happens with probability $1 - \frac{k}{n}$, we have
		\[
		\left\|f_i^{|p}(\distrD_i | p) - f_i^{|p}(\distrD_i^C | p)  \right\| = 0,
		\]
		as $\distrD_i^C | p = \distrD_i | p$.
		
		\item When $i \in C$, which happens with probability $\frac{k}{n}$, by Lemma~\ref{lm:close-subset-full}, we have
		\[
		\Ex_{C' \sim \distrS_{k-1}^{[n] \setminus \{i\} } } \left[ \left\| f_i^{|p}(\distrD_i | p) - f_i^{|p}(\distrD_i^{C' \cup \{i\}} | p)  \right\| \right] \le O\left(k \cdot \sqrt{\frac{j +\log n}{n}}\right).
		\]
	\end{itemize}
	
	Putting them together, we have
	\[
	\Ex_{p \sim \Prand^{(t-1)}} \Ex_{C \sim \distrS_{k}} \left[ \left\|f_i^{|p}(\distrD_i | p) - f_i^{|p}(\distrD_i^C | p)  \right\| \right]
\le 1/n^2 + \frac{k}{n} \cdot O\left(k \cdot \sqrt{\frac{j+\log n}{n}}\right),
	\]
	which proves \eqref{eq:ex-distance-small} for $t$.
\end{proofof}

Finally, we prove Claim~\ref{claim:large-D-p-clique}.
\newcommand{\FCon}{F}

\begin{proofof}{Claim~\ref{claim:large-D-p-clique}}

	Let $t_1,t_2,\dotsc,t_{\ell}$ be the indices of all previous $\ell$ turns with processor $i$ broadcasting, before the current $t$-th turn. We have $\ell \le j$.
	Let $x \in \{ z : z \in \{0,1\}^n, z_i = 0 \}$, note that $x$ is consistent with transcript $p$, if for all $a \in [\ell]$, we have
	\[
	f_i^{|p^{(t_a - 1)}}(x) = p_{t_a},
	\]
	where $p^{(t_a - 1)}$ denotes the first $t_a - 1$ bits of $p$. We set $\FCon_i(x,p) = 1$ if $x$ and $p$ are consistent, and $0$ otherwise.
	
	Consider the random process of generating $p \sim \Prand^{(t-1)}$, suppose inputs to all processors other than $i$ are fixed, let $x^{-i} = (x_1,x_2,\dotsc,x_{i-1},x_{i+1},\dotsc,x_{n}) \in \{0,1\}^{(n-1) \times n}$ be those fixed input. Let $P_{x^{-i}}^{(t)}$ be the distribution of the transcript when $x_i \sim \distrA_{\sf rand}^{[i]}$, and all other processors get (fixed) input according to $x^{-i}$. 
	
	For a fixed $x^{-i}$, note that there are only $2^\ell$ possible transcripts $p$ from $P_{x^{-i}}^{(t-1)}$, as the transcript is determined after fixing the output of processor $i$ at all $\ell$ rounds.	Therefore, let $T(x^{-i},x_i)$ be the transcript when all processors get inputs according to $x^{-i}$ and $x_i$, we can see when $p \sim P_{x^{-i}}^{(t-1)}$, $\FCon_i(x,p) = 1$ if and only if $T(x^{-i},x) = p$. That is,
	\[
	P_{x^{-i}}^{(t-1)}(p) = \Pr_{x_i \sim \distrA_{\sf rand}^{[i]}}[T(x^{-i},x_i) = p] = D_p^{(t-1)} / 2^{n - 1}.
	\]
	
	In above $P_{x^{-i}}^{(t-1)}(p)$ is the probability that getting $p$ from distribution $P_{x^{-i}}^{(t-1)}$. Then we have
	\begin{align*}
	&\Pr_{p \sim P_{x^{-i}}^{(t-1)}} \left[ D_p^{(t-1)} < 2^{-j-3\log n} \cdot 2^{n}  \right] \\
	=&\Pr_{p \sim P_{x^{-i}}^{(t-1)}} \left[ P_{x^{-i}}^{(t-1)}(p) < 2^{-j - 3\log n + 1}  \right] \\
	\le&2^{-j - 3\log n + 1} \cdot 2^{\ell} = 1/n^2.
	\end{align*}
	The last inequality holds because the support size of $P_{x^{-i}}^{(t-1)}$ is at most $2^{\ell}$ and $\ell \le j$.
	
	Hence, we have
	\begin{align*}
	\Pr_{p \sim \Prand^{(t-1)}} \left[D_p^{(t-1)} < 2^{-j-3\log n} \cdot 2^{n} \right] &= \Ex_{ x^{-i} \sim \Arand^{[-i]} } \left[ \Pr_{p \sim P_{x^{-i}}^{(t-1)}} \left[D_p^{(t-1)} < 2^{-j-3\log n} \cdot 2^{n} \right] \right]\\
	&\le 1 / n^2.
	\end{align*}
	In above $\Arand^{[-i]}$ denotes the marginal distribution of $\Arand$ on all rows except the $i$-th row.
\end{proofof}

\subsection{Proof for Lemma~\ref{lm:close-subset-full}}

We need the following lemma first, which is proved using tools from information theory.

\begin{lemma}\label{lm:close-subset-one}
	Let $n, t, k$ be integers such that $t,k \le n/10$, $D$ be a subset of $\{0,1\}^n$ with $|D| \ge 2^{n-t}$, $\distr_{D}$ be the uniform distribution on $D$, and $\distr_{D}^{[i]}$ be the uniform distribution on $\{ x : x \in D \text{ and } x_i = 1\}$, for all function $f : D \to \{0,1\}$, we have
	\[
	\Ex_{ i \leftarrow [n] } \left[ \left\| f(\distr_D) - f(\distr_D^{[i]}) \right\| \right] \le O\left(\sqrt{\frac{t}{n}}\right).
	\]
\end{lemma}
\begin{proof}
	Let $D^{[i]} := \{ x : x \in D \text{ and } x_i = 1\}$. Throughout the proof we will assume $X$ is a random variable drawn uniformly from $D$. For $i \in [n]$, let $X_i$ be the random variable of the $i$-th bit of $X$.
	
	We have $ \frac{|D^{[i]}|}{|D|} = \Pr[X_i = 1]$. By the sub-additivity of entropy, it follows that $\sum_{i=1}^{n} H(X_i) \ge H(X) \ge n - t$.
	
	That is, $\Ex_{i \leftarrow [n]} [1 - H(X_i)] = \frac{t}{n}$. By a simple Markov's inequality, with probability at least $1 - \frac{2t}{n}$ over $i \leftarrow [n]$, we have $H(X_i) \ge 1/2$. Note that $H(X_i) \ge 1/2$ implies $\Pr[X_i = 1] \ge 0.1$.
	
	Also, 
	\[
	I(X_i;f(X)) = H(X_i) - H(X_i | f(X)) \le 1 - H(X_i | f(X)).
	\]
	
	And by the sub-additivity of conditional entropy, we have
	\[
	\sum_{i=1}^{n} H(X_i | f(X)) \ge H(X | f(X)) \ge n-t-1.
	\]
	Therefore,
	\[
	\sum_{i=1}^{n} I(X_i;f(X)) \le n - (n-t-1) \le t + 1,
	\]
	or equivalently,
	\[
	\Ex_{i \leftarrow [n]} I(X_i;f(X)) \le \frac{t+1}{n}.
	\]
	Note that by Fact~\ref{ft:KL-mutual},
	\[
	I(X_i ; f(X)) := \Ex_{x \sim X_i} D( f(X)_{X_i = x} || f(X) ).
	\]
	Taking expected values over $i$ of both sides and using $\Ex_{i \leftarrow [n]} I(X_i;f(X)) \le \frac{t+1}{n}$ gives
	\[		
	\Ex_{i \leftarrow [n]} I(X_i;f(X)) = \Ex_{i \leftarrow [n]} \Ex_{x \sim X_i} D( f(X)_{X_i = x} || f(X) ) \le \frac{t+1}{n}.
	\]
	By Pinsker's inequality (Lemma~\ref{lm:pinsker}) and the fact that $\sqrt{x}$ is a concave function, we have
	\[
	\Ex_{i \leftarrow [n]} \Ex_{x \sim X_i} \| f(X)_{X_i = x} - f(X) \| \le \sqrt{\frac{t+1}{n}},
	\]
	and
	\[
	\Ex_{i \leftarrow [n]} \Pr[X_i = 1] \cdot \| f(X)_{X_i = 1} - f(X) \| \le \sqrt{\frac{t+1}{n}}.
	\]		
	
	Finally, note that with probability at least $1 - \frac{2t}{n}$ over $i \leftarrow [n]$, we have $\Pr[X_i = 1] \ge 0.1$. Putting everything together, we have
	\[		
	\Ex_{i \leftarrow [n]} \| f(X)_{X_i = 1} - f(X) \| \le \frac{2t}{n} + 10 \cdot \sqrt{\frac{t+1}{n}} \le O\left(\sqrt{\frac{t}{n}}\right). \qedhere
	\]
\end{proof}

Now we are ready to prove Lemma~\ref{lm:close-subset-full} (restated below).

\newcommand{\distrT}{\mathcal{T}}	

\begin{reminder}{Lemma~\ref{lm:close-subset-full}}
	Let $n, t, k$ be integers such that $t,k \le n^{1/4}$ and $t \ge 10\log n$, $D$ be a subset of $\{0,1\}^n$ with $|D| \ge 2^{n-t}$, $\distr_{D}$ be the uniform distribution on $D$, and $\distr_{D}^{C}$ be the uniform distribution on $\{ x : x \in D, x_i = 1 \text{ for all $i \in C$} \}$. For all function $f : D \to \{0,1\}$, we have
	\[
	\Ex_{C \sim \distrS_{k}^{[n]}} [ \| f(\distr_D) - f(\distr_D^{C}) \| ] \le O\left(k \cdot \sqrt{\frac{t}{n}}\right).
	\]
\end{reminder}
\begin{proofof}{Lemma~\ref{lm:close-subset-full}}
		
	Instead of choosing $C$ from $\distrS_k^{[n]}$, we choose an ordered $k$-tuple of $a = (a_1,a_2,\dotsc,a_k)$ of $k$ distinct elements in $[n]$ uniformly at random. Let the distribution be $\distrT_{k}^{[n]}$. 
	
	We have
	\begin{align}
	\Ex_{C \sim \distrS_{k}^{[n]}} [ \| f(\distr_D) - f(\distr_D^{C}) \| ] =& \Ex_{a \sim \distrT_{k}^{[n]}} [ \| f(\distr_D) - f(\distr_D^{ \{a_i\}_{i=1}^{k} }) \| ] \notag \\
	\le& \sum_{\ell=1}^{k} \Ex_{a \sim \distrT_{\ell}^{[n]}} [ \| f(\distr_D^{ \{a_i\}_{i=1}^{\ell-1} }) - f(\distr_D^{ \{a_i\}_{i=1}^{\ell} }) \| ] \notag \\
	\le& \sum_{\ell=0}^{k-1} \Ex_{a \sim \distrT_{\ell}^{[n]}} \Ex_{j \leftarrow [n] \setminus \{a_i\}_{i=1}^{\ell} } [ \| f(\distr_D^{ \{a_i\}_{i=1}^{\ell} }) - f(\distr_D^{ \{a_i\}_{i=1}^{\ell} \cup \{j\} }) \| ] \label{line:eq-term}
	\end{align}
	
	Now we are to bound the right side of \eqref{line:eq-term} for each $0 \le \ell \le k-1$ separately. For a subset $S \subseteq [n]$, let $D^{S} = \{ x : x \in D \wedge x_{i} = 1 \text{ for all $i \in S$} \}$.
	
	We first show with high probability, for $a \sim \distrT_{\ell}^{[n]}$, we have $D^{ \{a_i\}_{i=1}^{\ell} }$ is very large. The proof of the following claim is deferred to end of the whole proof.
	
	\begin{claim}\label{claim:DS-whp-large}
		For an integer $\ell \le n^{1/4}$,
		\[
		\Pr_{a \sim \distrT_{\ell}^{[n]}} [ | D^{ \{a_i\}_{i=1}^{\ell} } | \ge 2^{(n - \ell) - 3t} ] \ge 1 - O\left(\frac{t \cdot \ell}{n}\right).
		\]
	\end{claim}
	
	Now, by Claim~\ref{claim:DS-whp-large} and Lemma~\ref{lm:close-subset-one}, with probability at least $1 - O\left(\frac{t \cdot \ell}{n}\right)$ over $a \sim \distrT_{\ell}^{[n]}$, we have
	\[
	\Ex_{j \leftarrow [n] \setminus \{a_i\}_{i=1}^{\ell} } [ \| f(\distr_D^{ \{a_i\}_{i=1}^{\ell} }) - f(\distr_D^{ \{a_i\}_{i=1}^{\ell} \cup \{j\} }) \| \le O\left(\sqrt{\frac{3t}{n-\ell}}\right) = O\left( \sqrt{\frac{t}{n}} \right).
	\]
	
	Putting them together, we have
	\begin{align*}
	\Ex_{a \sim \distrT_{\ell}^{[n]}} \Ex_{j \leftarrow [n] \setminus \{a_i\}_{i=1}^{\ell} } [ \| f(\distr_D^{ \{a_i\}_{i=1}^{\ell} }) - f(\distr_D^{ \{a_i\}_{i=1}^{\ell} \cup \{j\} }) \| ] 
	\le O\left( \frac{t \cdot \ell}{n} + \sqrt{\frac{t}{n}} \right)
	\end{align*}
	
	Summing everything up for $0 \le \ell \le k - 1$, we have
	\[		
	\Ex_{C \sim \distrS_{k}^{[n]}} [ \| f(\distr_D) - f(\distr_D^{C}) \| ] \le O\left( k^2 \cdot \frac{t}{n} + k \sqrt{\frac{t}{n}}\right) = O\left(k \sqrt{\frac{t}{n}}\right).
	\]
	
\end{proofof}

Finally, we prove Claim~\ref{claim:DS-whp-large}, which is the most technical proof of this section.

\newcommand{\Good}{\textsf{Good}}
\newcommand{\seqa}[1]{\{a_i\}_{i=1}^{#1}}
\newcommand{\Edge}[2]{E_{\seqa{#1} \rightarrow #2}}

\begin{proofof}{Claim~\ref{claim:DS-whp-large}}
	
	We begin with some notations.
	
	\paragraph*{Subset Tree.}
	
	We can view the process of choosing the $k$-tuples as growing a tree. For each $0 \le \ell \le k$ and each sequence $\seqa{\ell}$ from $\distrT_{\ell}^{[n]}$, we build a tree node $T_{\seqa{\ell}}$. For each $j \in [n] \setminus \seqa{\ell}$, we say node $T_{\seqa{\ell} \cup \{j\}}$ is a child of the node $T_{\seqa{\ell}}$, and denote the edge between them as $\Edge{\ell}{j}$. With this interpretation, the process of choosing $a \sim \distrT_{k}^{[n]}$ can be seen as starting from the root $T_{\emptyset}$, and descending to a random child for $k$ times.
	
	We also define
	\[
	Z_{\seqa{\ell}} = (n - \ell) - \log_2 |D^{\seqa{\ell}}|,
	\]
	and
	\[
	Y_{\seqa{\ell}} = Z_{\seqa{\ell}} - Z_{\seqa{\ell-1}}.
	\]
	
	That is, $Z_{\seqa{\ell}}$ is the gap between the entropy of the set corresponding to the node and the ``full entropy'' $n-\ell$, while $Y_{\seqa{\ell}}$ is the increase of that entropy gap on its parent.
	
	Note that the claim asks to upper bound
	\[
	\Pr_{a \sim T_{\ell}^{[n]}} [Z_{\seqa{\ell}} > 3t ],
	\]
	and we have
	\[
	Z_{\emptyset} = t.
	\]
	
	\paragraph*{Good Nodes, Good Edges, Bad Nodes, Bad Edges, and Edge Labels.} We next define when a node (or an edge) is good or bad. The root $T_{\emptyset}$ is a good node. If the parent of the node is a bad node then it is also a bad node. If a node is a bad node, then all edges in its sub-tree are bad edges.
	
	If $T_{\seqa{\ell}}$ is a good node, we look at all $j \in [n] \setminus \seqa{\ell}$. We say the edge from $T_{\seqa{\ell}}$ to $T_{ \seqa{\ell} \cup \{j\} }$, denoted as $E_{\seqa{\ell} \rightarrow j}$ is a good edge, if
	\[
	H_{X \sim D^{\seqa{\ell}}}(X_j) \ge 0.9,
	\]
	otherwise it is a bad edge. The above basically guarantees to us that a large enough subset (at least a constant fraction) of $D^{\seqa{\ell}}$ contains a $1$ as its $j$th entry.
	
	We mark $T_{\seqa{\ell} \cup \{j\}}$ as a bad node if $E_{\seqa{\ell} \rightarrow j}$ is a bad edge, or $Z_{ \seqa{\ell} \cup \{j\} } > 3t$, otherwise it is a good node.
	
	For a good edge $\Edge{\ell}{j}$, we say it has label $k$ if $|Y_{\seqa{\ell} \cup \{j\}}| \in (2^{-k},2^{-k+1}]$. 
	
	Since it is an good edge, we have
	\[
	H_{X \sim D^{\seqa{\ell}}}(X_j) \ge 0.9,
	\]
	and by Fact~\ref{ft:H-p-property}, it follows
	\[
	\Pr_{X \sim D^{\seqa{\ell}}}[X_j = 1] \ge 0.3.
	\]
	Therefore,
	\[
	|D^{\seqa{\ell} \cup \{j\}}| \ge 0.3 \cdot |D^{\seqa{\ell}}|.
	\]
	
	Therefore, $Z^{\seqa{\ell} \cup \{j\}} \le Z^{\seqa{\ell}} + \log(1/0.3) - 1 \le Z^{\seqa{\ell}} + 1$, which means $Y_{\seqa{\ell} \cup \{j\}} \le 1$. Hence, a good edge's label is at least $1$.
	\paragraph*{Basic Facts.} We need the following two basic facts, whose proof can be found at the end of the proof.
	
	\begin{ft}\label{ft:whp-good-edge}
	Let $T_{\seqa{\ell}}$ be a good node, we have
	\[
	\Pr_{j \in [n] \setminus \seqa{\ell}} [E_{\seqa{\ell} \rightarrow j} \text{ is good}] \ge  1 - O\left(\frac{t}{n} \right).
	\]
	\end{ft}

	\begin{ft}\label{ft:bound-on-label-k}
	Let $T_{\seqa{\ell}}$ be a good node and $k$ be an integer, we have
	\[
	\Pr_{j \in [n] \setminus \seqa{\ell}} [E_{\seqa{\ell} \rightarrow j} \text{ has label $k$}] \le O\left( \frac{4^k \cdot t}{n} \right).
	\]
	\end{ft}

	\paragraph*{The Bound.}
	
	\newcommand{\eventgood}{\event_{\sf good}}
	\newcommand{\eventbad}{\event_{\sf bad}}
	
	Now we are going to lower bound the probability of the event that all nodes $T_{\seqa{d}}$ for $0 \le d \le \ell$ are good, denoted as event $\eventgood$. Clearly this provides a lower bound on $\Pr_{a \sim T_{\ell}^{[n]}} [Z_{\seqa{\ell}} > 2t ]$.
	
	Suppose $\eventgood$ doesn't happen, let $d$ be the first index such that $T_{\seqa{d}}$ is a bad node, let this event be $\eventbad^{d}$. Clearly we have
	\[
	\Pr[\eventgood] = 1 - \sum_{d=0}^{\ell} \Pr[\eventbad^{d}].
	\]
	
	Therefore it suffices to provide an upper bound for each $\Pr[\eventbad^{d}]$, note that $\eventbad^{d}$ is defined as
	\[
	\left[\text{$T_{\seqa{j}}$ is good for all $0 \le j \le d - 1$ and $T_{\seqa{d}}$ is bad}\right].
	\]
	
	There are two possible cases, the first case is that the edge $\Edge{d-1}{a_{d}}$ is an bad edge, which happens with probability at most $O\left( \frac{t}{n} \right)$ by Fact~\ref{ft:whp-good-edge}.
	
	The second case is that the edge $\Edge{d-1}{a_{d}}$ is an good edge. In that case, by definition, we have $Z_{\seqa{d}} > 3t$, which also means
	
	\[
	\sum_{j=1}^{d} Y_{\seqa{j}} > 2t.
	\]
	
	Let $N_{k}$ be the number of edges in the path $\{\Edge{j-1}{a_j}: j \in [d]\}$ with label $k$. %\onote{Make sure the sentence above is correct, because I changed its meaning to what I thought you meant to say} We have
	\[
	\sum_{k=1}^{\infty} N_{k} \cdot 2^{-k+1} > 2t,
	\]
	which simplifies to
	\[
	\sum_{k=1}^{\infty} N_{k} \cdot 2^{-k} > t.
	\]
	
	Note that since $N_{k} \le d$, we have
	\[
	\sum_{k=\log_2 (2d/t) + 1}^{\infty} N_k \cdot 2^{-k} \le d \cdot \frac{t}{2d} \le \frac{t}{2}.
	\]
	
	Which means
	\[
	\sum_{k=1}^{\log_2 (2d/t)} N_{k} \cdot 2^{-k} > t/2.
	\]
	
	In particular, this means there exists an $k \in [\log_2 (2d/t)]$, such that
	\[
	N_{k} \cdot 2^{-k} > \frac{t}{2 \log n} \Rightarrow N_k \ge \frac{2^k \cdot t}{2 \log n}.
	\]
	
	Let the above be event $\eventbad^{d,k}$, we have
	\[
	\Pr[\eventbad^{d}] \le \sum_{k=1}^{\log (2d/t)} \Pr[\eventbad^{d,k}].
	\]
	
	And by Fact~\ref{ft:bound-on-label-k}, we have
	\begin{align*}
	\Pr[\eventbad^{d,k}] &\le O\left( \frac{4^k \cdot t}{n} \right)^{ \frac{2^k \cdot t}{2 \log n} } \cdot \binom{d}{\frac{2^k \cdot t}{2 \log n}} \\
				   &\le O\left( \frac{4^k \cdot t}{n} \cdot d \right)^{\frac{2^k \cdot t}{2 \log n}}.
	\end{align*}
	
	Note that $4^k \le (2d/t)^2 = O(d^2/t^2)$, $d \le \ell \le n^{1/4}$ and $t \ge 10\log n$, the above simplifies to
	\[
	\Pr[\eventbad^{d,k}] \le O\left( \frac{d^3/t}{n} \right)^{\frac{2^k \cdot t}{2 \log n}} \le n^{-1/4 \cdot 10} \le n^{-2}.
	\]
	
	Putting everything together, we have
	\[
	\Pr[\eventbad^{d}] \le \log n \cdot n^{-2} + O\left(\frac{t}{n}\right) = O\left(\frac{t}{n}\right),
	\]
	and
	\[
	\Pr[\eventgood] \ge 1 - \ell \cdot \left( \frac{t}{n} \right) \ge 1 - O\left( \frac{t \cdot \ell}{n} \right).
	\]
	The above completes the proof.
\end{proofof}

Now we finish the whole proof by proving Fact~\ref{ft:whp-good-edge} and Fact~\ref{ft:bound-on-label-k}.

\begin{reminder}{Fact~\ref{ft:whp-good-edge}}
Let $T_{\seqa{\ell}}$ be a good node. We have
\[
\Pr_{j \in [n] \setminus \seqa{\ell}} [E_{\seqa{\ell} \rightarrow j} \text{ is good}] \ge  1 - O\left(\frac{t}{n} \right).
\]
\end{reminder}

\begin{reminder}{Fact~\ref{ft:bound-on-label-k}}
	Let $T_{\seqa{\ell}}$ be a good node and $k$ be an integer. We have
	\[
	\Pr_{j \in [n] \setminus \seqa{\ell}} [E_{\seqa{\ell} \rightarrow j} \text{ has label $k$}] \le O\left( \frac{4^k \cdot t}{n} \right).
	\]
\end{reminder}

\begin{proofof}{Fact~\ref{ft:whp-good-edge} and Fact~\ref{ft:bound-on-label-k}}
	Let $X \sim D^{\seqa{\ell}}$, we have $H(X) \ge n - \ell - 2t$. Also, since $X_j$ for $j \in \seqa{\ell}$ is always $1$, and therefore has entropy $0$, by the sub-additive of entropy, we have
	\[
	\sum_{j \in [n] \setminus \seqa{\ell}} H(X_j) \ge n - \ell - 2t.
	\]
	Or equivalently, we have
	\[
	\Ex_{j \in [n] \setminus \seqa{\ell}} (1 - H(X_j)) \le \frac{2t}{n-\ell} \le \frac{4t}{n}.
	\]
	
	By a simple Markov's inequality, we have
	\[
	\Pr_{j \in [n] \setminus \seqa{\ell}} \left[1 - H(X_j) \ge 0.1 \right] \le O\left(\frac{t}{n}\right),
	\]
	which proves Fact~\ref{ft:whp-good-edge}.
	
	Now, let $B_j$ be the set of $j$ satisfying $H(X_j) < 0.9$. We have
	\[
	\sum_{j \in [n ]\setminus (\seqa{\ell} \cup B_j )} (1 - H(X_j)) \le 2t.
	\]
	
	Now, let $p_j = \Pr[X_j = 1]$ and $z_j = (p_j - 1/2)$. For $j \in [n ]\setminus (\seqa{\ell} \cup B_j )$, we have $H(p_j) \ge 0.9$ and $|z_j| \le 0.2$ from Fact~\ref{ft:H-p-property}, and also
	\begin{equation}\label{eq:bound-on-sum}
	\sum_{j \in [n ]\setminus (\seqa{\ell} \cup B_j )} 2 \cdot z_j^2 \le 2t.
	\end{equation}
	
	Also, by definition, we have
	\[
	Y_{\seqa{\ell} \cup \{j\}} = \log(1/p_j) - 1 = - \log(2p_j) = - 2 \log(1 + 2z_j).
	\]
	
	Consider the function $g(z) := \log(1+z) / z$, we can see it is a decreasing function when $z \in [-0.4,0.4]$, and we have
	\[
	 0.8\le g(0.4)\le \frac{\log(1 + z)}{z} \le g(-0.4) \le 1.3.
	\]
	
	Therefore, if the edge $\Edge{\ell}{j}$ has label $k$, we know that
	\[
	|2 \log(1 + 2 z_j)| \ge 2^{-k} \Rightarrow |6 z_j| \ge 2^{-k} \Rightarrow z_j^2 \ge  4^{-k} \cdot \frac{1}{36}.
	\]
	
	Using \eqref{eq:bound-on-sum}, we see there are at most $O\left( 4^k \cdot t \right)$ $j$'s such that $\Edge{\ell}{j}$ has label $k$. From which Fact~\ref{ft:bound-on-label-k} follows directly.
\end{proofof}
\section{Proof Overview For the PRG Construction}
\label{sec:PRG-overview}

%\lnote{temporarily make this as a proof overview of the PRG Construction...}

In this section we provide an overview of the proof for our PRG construction. We first consider a toy example: a very simple PRG which constructs one pseudo-random bit, and fools any one-round \BCAST(1) protocol. Its proof already illustrates the key proof strategy which is used to prove our full PRG results. Then in Subsections \ref{generalize1} and \ref{generalize2} we sketch the key ideas to generalize the proof for the general PRG theorem.

\paragraph*{The Toy PRG.} Here we describe the PRG, and below we will analyze it to show it is indeed pseudo-random. Suppose there are $n$ processors, and each processor receives $k$ \emph{truly random bits}. Suppose there is also a shared random bit-vector $b$ of length $k$, which is also sampled uniformly at random. Then each processor's extra pseudo-random bit is  the inner product (modulo 2) of the vector formed by its random bits and $b$ (so the complete pseudo-random string is its initial $k$ random bits concatenated with these extra random bits obtained with the inner product). Note that in the typical case $n \gg k$, this PRG generates $n$ pseudo-random bits out of a shared random string $b$ of length $k$. When analyzing the PRG, we think of $b$ as a ``secret'' string, since distinguishing the PRG from true randomness corresponds to discovering whether such a $b$ exists.

The goal here is to show that the above PRG construction and the case that all processor get $k+1$ truly random bits are indistinguishable to a one-round \BCAST(1) protocol (see Theorem~\ref{theo:one-round} for a formal statement). We begin with some notations.

\paragraph*{Notations.} 

Throughout the paper, except when explicitly stated, all matrices and vectors are over $\mathbb{F}_2$. We use $\{0,1\}^{n}$ ($\{0,1\}^{n \times m}$) and $\mathbb{F}_2^{n}$ ($\mathbb{F}_2^{n \times m}$) interchangeably. For two vectors $u$ and $v$, we use $(u,v)$ to denote their concatenation.

Recall that we can assume all processors are deterministic as we are trying to prove a lower bound for distinguishing two input distributions by Yao's principle. Processor $i$ can then be defined by a function $f_i : \{0,1\}^{k+1} \times \{0,1\}^{*} \to \{0,1\}$, such that $f_i(z,p)$ is the bit that player $i$ outputs when it gets the input $z$ and transcript $p$. We use $f_i^{|p}$ to denote the function $f_i(\cdot,p)$ for simplicity. If transcript $p$ is incompatible with player $i$ having input $z$, then we set $f_i(z, p)$ arbitrarily.

We use $\distr_{[b]}$ to denote the uniform distribution on the set $\{ (x,x \cdot b) : x \in \{0,1\}^{k} \}$, which is the distribution of inputs a processor receives when the shared random string during the construction of the pseudo-randomness is $b$. We can now formally state our theorem.
\begin{theo}\label{theo:one-round}
	Let $k$ be an integer and $n$ be the number of processors. Consider the following two cases:
	
	\begin{itemize}
		\item (A) All processors receive random inputs from $\distr_{k+1}$.
		\item (B) Let $b$ be a uniform sample from $\distr_{k}$, then all processors receive inputs from $\distr_{[b]}$.
	\end{itemize}
	
	%Let $U'$ be the distribution $\sum_{b \in \mathbb{F}_2^k} \frac{1}{2^k} U_{[b]}$, where $U_{[b]}$ is the distribution where processor $i$ receives a random vector $a_i \in \mathbb{F}_2^k$, along with $a_i \cdot b$.
	
	For any one-round \BCAST(1) protocol, the statistical distance between the distributions of its transcripts in case (A) and (B) is at most $O\left(\frac{n}{2^{k/2}} \right)$.
\end{theo}

We need the following technical lemma, whose proof can be found at the end of this section.

\begin{lemma}\label{lm:fourier}
	Given a function $f : \{0,1\}^{k+1} \to \{0,1\}$, we have
	\[
	\sum_{b \in \{0,1\}^{k}} \| f(\distr_{k+1}) - f(\distr_{[b]}) \|^2 \le \Ex_{x \sim \distr_{k+1}}[f(x)] \le 1.
	\]	
	Note that $f(\distr_{k+1})$ and $f(\distr_{[b]})$ are two distributions on $\{0,1\}$.
\end{lemma}

Intuitively, the above lemma says that for any function $f$ (think of this as a function describing a processor), it cannot distinguish distributions $\distr_{[b]}$ and $\distr_{k+1}$ for most strings $b$. So, fixing a few random entries of $x$ to be $1$ doesn't change the probability that $f(x)$ is $1$ by much.

Now we are ready to prove Theorem~\ref{theo:one-round}.

\begin{proof}
	Instead of viewing the algorithm as a single round algorithm, we will prove a slightly stronger lower bound. Consider the model where we have $n$ turns. On the $t^{th}$ turn, processor $t$ gets to send a single bit. This model is stronger than one round of the \BCAST(1) model, since it allows the later processors to condition their outputs on earlier the processors' messages. Hence, our lower bound implies a lower bound for the \BCAST(1) model as well.
	
	\paragraph*{Notations.} Let $\Prand^{(t)}$ and $\distrP_{[b]}^{(t)}$ be the distributions of the transcript of the first $t$ turns when all processors get a random input from $\distr_{k+1}$ and $\distr_{[b]}$ respectively.
	
	Note that to prove the theorem, it suffices to show that the distribution $\Prand^{(n)}$ is close to $\distrP_{[b]}^{(n)}$ for most choices of $b$. For this purpose, we are going to prove the following inequality holds for any $t \le n$:
	
	\begin{equation}\label{eq:ex-distance-small-toy}
	\Ex_{b \sim \distr_k} \left[ \| \Prand^{(t)} - \distrP_{[b]}^{(t)} \| \right] \le t \cdot 2^{-k/2}.
	\end{equation}
	
	It is easy to see that plugging in $t = n$, \eqref{eq:ex-distance-small-toy} implies the theorem. To prove \eqref{eq:ex-distance-small-toy} for all $t$, we induct on $t$. Clearly, \eqref{eq:ex-distance-small} holds when $t = 0$. So it suffices to show that when it holds for $t-1$, it also holds for $t$.
	
	For $b \in \{0,1\}^k$, we wish to bound the distance $\| \Prand^{(t)} - \distrP_{[b]}^{(t)} \|$. By Lemma~\ref{lm:dist}, it follows that
	\begin{equation}
	\| \Prand^{(t)} - \distrP_{[b]}^{(t)} \| \le \| \Prand^{(t-1)} - \distrP_{[b]}^{(t-1)} \| + \Ex_{p \sim \Prand^{(t-1)}} \left[ \left\|f_t^{|p}(\distr_{k+1}) - f_t^{|p}(\distr_{[b]})  \right\| \right].
	\label{eq:bound-next-t}
	\end{equation}
	
	Recall that in above $f_t^{|p}$ is the output function of process $t$ when seeing the transcript $p$.
	
	For each $b \in \{0,1\}^k$ and transcript $p \in \{0,1\}^{t-1}$, we define scores $s_{b,p}$ and $s_b$ as follows:
	\[
	s_{b,p} := \left\|f_t^{|p}(\distr_{k+1}) - f_t^{|p}(\distr_{[b]})  \right\| \quad\text{and}\quad s_b := \Ex_{p \sim \Prand^{(t-1)}} [s_{b,p}].
	\]
	
	It suffices to give an upper bound on $\Ex_{b \sim \distr_k}[s_b]$. By Lemma~\ref{lm:fourier}, for all $p \in \{0,1\}^{t-1}$, we have
	\[
	\sum_{b \in \{0,1\}^k} s_{b,p}^2 \le 1,
	\]
	and therefore
	\[
	\sum_{b \in \{0,1\}^k} s_{b,p} \le 2^{k/2} \quad\text{and}\quad \Ex_{b \sim \distr_{k}} [s_{b,p}] \le 2^{-k/2}.
	\]
	By the definition of $s_b$, it follows that
	\[
	\Ex_{b \sim \distr_k} [s_b] \le 2^{-k/2}.
	\]
	
	Therefore, we have
	
	\begin{align*}
	\Ex_{b \sim \distr_k} \left[ \| \Prand^{(t)} - \distrP_{[b]}^{(t)} \| \right] &\le 
	\Ex_{b \sim \distr_k} \left[ \| \Prand^{(t-1)} - \distrP_{[b]}^{(t-1)} \| + \Ex_{p \sim \Prand^{(t-1)}} \left[ \left\|f_t^{|p}(\distr_{k+1}) - f_t^{|p}(\distr_{[b]})  \right\| \right] \right]\\
	\le&(t-1) \cdot 2^{-k/2} + \Ex_{b \sim \distr_k} [s_b]\\
	\le&t \cdot 2^{-k/2}.
	\end{align*}
	
	The above proves \eqref{eq:ex-distance-small-toy} for $t$, which completes the whole proof.
	
	%\textbf{To add an explanation for the first inequality.}
	
\end{proof}

Finally, we prove Lemma~\ref{lm:fourier}.
\begin{proofof}{Lemma~\ref{lm:fourier}}
	
	Note that since $f$ is Boolean valued, we have
	\[
	\| f(\distr_{k+1}) - f(\distr_{[b]}) \| = \left| \Ex_{x \sim \distr_{k+1}}[f(x)] - \Ex_{x \sim \distr_{[b]}}[f(x)] \right|.
	\]
	
	The proof is an application of the analysis of Boolean functions (see Section~\ref{sec:analysis-of-boolfunc}). Let $b \in \{0,1\}^{k} $, and let $S_b$ be the corresponding subset of $[k]$ (if $b_i = 1$ then $i \in S_b$). We use $\odistr_{[b]}$ to denote the uniform distribution on the set $\{ (x,1-x\cdot b) : x \in \{0,1\}^{b} \}$, that is, the uniform distribution on the complement of the support of $\distr_{[b]}$. 
	
	Note that for every $x$ from the support of $\distr_{[b]}$, we have $x \cdot (b,1) = 0$, and for every $x$ from the support of $\odistr_{[b]}$, we have $x \cdot (b,1) = 1$. Then we have
	\begin{align*}
	\WH{f}(S_b \cup \{k+1\}) :=& \Ex_{x \sim \distr_{k+1}} \left[ f(x) \cdot (-1)^{(b,1) \cdot x} \right].\\
	=& \frac{1}{2} \cdot \left( \Ex_{x \sim \distr_{[b]}}[f(x)] - \Ex_{x \sim \odistr_{[b]}}[f(x)] \right)\\
	=& \frac{1}{2} \cdot \left( 2\Ex_{x \sim \distr_{[b]}}[f(x)] - \Ex_{x \sim \odistr_{[b]}}[f(x)] - \Ex_{x \sim \distr_{[b]}}[f(x)] \right)\\
	=& \Ex_{x \sim \distr_{[b]}}[f(x)] -\Ex_{x \sim \distr_{k+1}}[f(x)].
	\end{align*}
	
	By Parseval's identity (see Section~\ref{sec:analysis-of-boolfunc}) and the fact that $f$ is Boolean valued, we have
	\[
	\sum_{b \in \{0,1\}^{k}} \WH{f}(S_b \cap \{k+1\})^2 \le \Ex_{x \sim \distr_{k+1}} [f(x)^2] = \Ex_{x \sim \distr_{k+1}}[f(x)],
	\]
	and it follows that
	\[
	\sum_{b \in \{0,1\}^{k}} \left(\Ex_{x \sim \distr_{[b]}}[f(x)] -\Ex_{x \sim \distr_{k+1}}[f(x)]\right)^2 \le \Ex_{x \sim \distr_{k+1}}[f(x)] \le 1. \qedhere
	\] 
\end{proofof}

\subsection{Generalization to Multi-Round Case}\label{generalize1}
Now we outline how to extend the proof to the multi-round case. The PRG is still the same as the toy PRG, we just need to prove it also fools multiple round \BCAST(1) protocols (i.e. Theorem~\ref{theo:multi-round}). In the following, we will explain the key difficulty for generalizing the previous proof to the multi-round case, and how we address them. 

%\lnote{need to adjust}
%\lnote{make sure to talk about the distribution of the transcripts}
\begin{theo}\label{theo:multi-round}
	Consider the following two cases:
	
	\begin{itemize}
		\item (A) All processors receive random inputs from $\distr_{k+1}$.
		\item (B) Let $b$ be a uniform sample from $\distr_{k}$, then all processors receive inputs from $\distr_{[b]}$.
	\end{itemize}
	
	%Let $U'$ be the distribution $\sum_{b \in \mathbb{F}_2^k} \frac{1}{2^k} U_{[b]}$, where $U_{[b]}$ is the distribution where processor $i$ receives a random vector $a_i \in \mathbb{F}_2^k$, along with $a_i \cdot b$.
	
	For $j \le k/10$, and any $j$-round \BCAST(1) protocol, the statistical distance between the distributions of its transcripts in case (A) and (B) is at most $O\left( \frac{j \cdot n}{2^{k/9}} \right)$.
\end{theo}

The key technical part of the proof of Theorem~\ref{theo:one-round}, is to bound $\| \Prand^{(t)} - \distrP_{[b]}^{(t)} \|$, i.e., the Inequality~\eqref{eq:bound-next-t}:
\[
\| \Prand^{(t)} - \distrP_{[b]}^{(t)} \| \le \| \Prand^{(t-1)} - \distrP_{[b]}^{(t-1)} \| + \Ex_{p \sim \Prand^{(t-1)}} \left[ \left\|f_t^{|p}(\distr_{k+1}) - f_t^{|p}(\distr_{[b]})  \right\| \right].
\]

\newcommand{\Xrand}{X_{\textsf{rand}}}

Let $\Xrand$ ($X_{[b]}$) denote the random variable for the input to the processor $i$ broadcasting at the $t^{th}$ turn, in the case all processors receive inputs from $\distr_{k+1}$ ($\distr_{[b]}$). Inequality~\eqref{eq:bound-next-t} holds crucially because $\Xrand$ ($X_{[b]}$) is independent of the previous part of the transcript during the first $(t-1)$ turns ($b$ is fixed).

However, the independence condition no longer holds in the multi-round case, as the transcript contains previous broadcasts of the \emph{same} processor $i$, which contain information about processor $i$'s input. To deal with that, we have to consider the conditional random variables $\Xrand^{|p}$ and $X_{[b]}^{|p}$ which are $\Xrand$ and $X_{[b]}$ conditioning on seeing the transcript $p$. 

Let $D^{(t-1)}_p$ denote the set of inputs to $f_i$ which are consistent with the transcript $p \in \{0,1\}^{t-1} $,\footnote{That is, simulating $f_i$ with transcript $p$ on that input results in transcript $p$ itself.} then $\Xrand^{|p}$ and $X_{[b]}^{|p}$ distribute uniformly on $\{0,1\}^{k+1} \cap D^{(t-1)}_p$ and $\{ (x,x \cdot b) : x \in \{0,1\}^{k} \} \cap D^{(t-1)}_p$. We use $\distr_{k+1,p}$ and $\distr_{[b],p}$ to denote their distributions. Then we can state a bound similar to \eqref{eq:bound-next-t} in the multi-round case:
\[
\| \Prand^{(t)} - \distrP_{[b]}^{(t)} \| \le \| \Prand^{(t-1)} - \distrP_{[b]}^{(t-1)} \| + \Ex_{p \sim \Prand^{(t-1)}} \left[ \left\|f_t^{|p}(\distr_{k+1,p}) - f_t^{|p}(\distr_{[b],p})  \right\| \right].
\]

Our one-round proof depends on Lemma~\ref{lm:fourier}, which cannot be used directly to bound the right side of the above inequality. Luckily, we are able to generalize Lemma~\ref{lm:fourier} such that it works as long as $D^{(t-1)}_p$ is sufficiently large (see Lemma~\ref{lm:fourier2}), which happens to be the case with high probability (see Claim~\ref{claim:large-D-p}).

\subsection{Generalization to the Complete PRG}\label{generalize2}

Before discussing how to generalize the proof to get a complete PRG. We give a formal description of the PRG here.

\paragraph*{The Full PRG.} Suppose there are $n$ processors and the PRG wants to create $m$ pseudo-random bits that fool an $\Omega(k)$-round $\BCAST(1)$ protocol. Then the PRG is described as follows: each processor gets $k$ \emph{truly random bits}. There is also a hidden ``secret" matrix $M$ of size $k \times (m-k)$, which distributes uniformly random (when constructing the pseudo-randomness, this matrix is created by having each processor broadcast some additional uniformly random bits until there are enough to create the matrix). Then each processor's extra $m-k$ pseudo-random bits are simply the vector matrix product of its random bits and $M$, i.e., $x^{T} M$ (see also Theorem~\ref{theo:PRG-formal}). 
%Note that in the typical case $n \gg k$, this PRG generates $n$ pseudo-random bits out of a secret string $b$ of length $k$.

The generalization to the complete PRG case (Theorem~\ref{theo:PRG}) is quite technical.  To state the whole technical theorem, we need to introduce some definitions. Let $M \in \{0,1\}^{n \times m}$. We use $\distr_{M}$ to denote the uniform distribution on the following set $\{ (x,x^{T} M)  : x \in \{0,1\}^{n} \}$, which is a subset of $\{0,1\}^{n+m}$. For integers $n$ and $m$, we use $\distr_{n \times m}$ to denote the uniform distribution on $\{0,1\}^{n \times m}$. Formally, we want to show:

\begin{theo}\label{theo:PRG}
	Let $n,m,k$ be three integers. Consider the following two cases:
	
	\begin{itemize}
		\item (A) All processors receive random inputs from $\distr_{m}$.
		\item (B) Let $M$ be a uniform sample from $\distr_{k \times (m-k)}$, then all processors receive inputs from $\distr_{M}$.
	\end{itemize}
	
	%Let $U'$ be the distribution $\sum_{b \in \mathbb{F}_2^k} \frac{1}{2^k} U_{[b]}$, where $U_{[b]}$ is the distribution where processor $i$ receives a random vector $a_i \in \mathbb{F}_2^k$, along with $a_i \cdot b$.
	
	For $j \le k/10$, $m \le 2^{k/20}$ and any $j$-round \BCAST(1) protocol, the statistical distance between the distributions of its transcripts in case (A) and (B) is at most $O\left( \frac{j \cdot n}{2^{k/9}} \right)$.
\end{theo}

The proof strategy is still similar to that of Theorem~\ref{theo:multi-round}. But now for each turn $t$, we need to maintain a set $S^{(t)}$ of secret matrices $M \in \{0,1\}^{k \times (m-k)}$ instead of a set of secret strings. Following the same reasoning as in the previous subsection, we can state a similar bound in this case:
\[
\| \Prand^{(t)} - P_{M}^{(t)} \| \le \| \Prand^{(t-1)} - P_{M}^{(t-1)} \| + \Ex_{p \sim \Prand^{(t-1)}} \left[ \left\|f_t^{|p}(\distr_{m,p}) - f_t^{|p}(\distr_{M,p})  \right\| \right].
\]
In which we use $P_M^{(t)}$ to denote the distribution of the transcript of the first $t$ rounds when all processors get random input from $\distr_{M}$. And $\distr_{m,p}$ and $\distr_{M,p}$ are distributions to the current processor $i$ conditioning on seeing transcript $p$. Using a clever hybrid argument, we are able to prove the sufficient technical lemma (Lemma~\ref{lm:fourier4}) to bound the right side of the above inequality.
\section{Creating a Single Extra Pseudo-random Bit And an Average Case Lower Bound}
\label{sec:one-secret}

\newcommand{\distrhard}{\distrD_{\sf hard}}
\newcommand{\distrother}{\distrD_{\sf other}}

%In this section we show how the processors can each create one single extra pseudo-random bit (Theorem~\ref{theo:multi-round}, restated below). That is, each processor will generate $k$ random bits, and then using $k$ extra shared random bits, each processor will add a single bit to his random string, and we show that this collective distribution is pseudo-random.

%This section is split into two subsections. In the first subsection, we prove that the distribution is pseudo-random for a one-round protocol. Then, we prove the distribution is pseudo-random for a multi-round protocol. The one-round proof is substantially simpler, and serves as a good overview for some of the ideas for the more general proof.

\begin{comment}
\begin{theo}
Consider the following distribution $D$:
sample a random vector $b \in \mathbb{F}_2^k$, and each processor gets a random $a \in \mathbb{F}_2^k$ and $a \cdot b$. Then it takes $\Omega(k)$ (up to polylogs?) rounds for the processors to distinguish $D$ from the uniform distribution with probability at least $\frac{2}{3}$.
\end{theo}
\end{comment}

In this section we show our toy PRG (see Section~\ref{sec:PRG-overview}) also fools multiple rounds \BCAST(1) protocols by proving Theorem~\ref{theo:multi-round} (restated below). We also show that our average case lower bound (Theorem~\ref{theo:average-lowb}) is a simple corollary of it.

\begin{reminder}{Theorem~\ref{theo:multi-round}.}
	Consider the following two cases:
	
	\begin{itemize}
		\item (A) All processors receive random inputs from $\distr_{k+1}$.
		\item (B) Let $b$ be a uniform sample from $\distr_{k}$, then all processors receive inputs from $\distr_{[b]}$.\footnote{recall that $\distr_{[b]}$denotes the uniform distribution on the set $\{ (x,x \cdot b) : x \in \{0,1\}^{k} \}$}
	\end{itemize}
	
	%Let $U'$ be the distribution $\sum_{b \in \mathbb{F}_2^k} \frac{1}{2^k} U_{[b]}$, where $U_{[b]}$ is the distribution where processor $i$ receives a random vector $a_i \in \mathbb{F}_2^k$, along with $a_i \cdot b$.
	
	For $j \le k/10$, and any $j$-round \BCAST(1) protocol, the statistical distance between the distributions of its transcripts in case (A) and (B) is at most $O\left( \frac{j \cdot n}{2^{k/9}} \right)$.
\end{reminder}

\subsection{An Average Case Lower Bound for \BCAST(1)}

First, we show Theorem~\ref{theo:multi-round} implies the average case lower bound we want.

\begin{reminder}{Theorem~\ref{theo:average-lowb}}
	Let $n$ be a large enough integer and $\Ffullrank : \{0,1\}^{n \times n} \to \{0,1\}$ be the indicator function that whether the given matrix has full rank. Suppose there are $n$ processors, $i$-th processor is given with the $i$-th row of the input matrix. For all $n/20$-round \BCAST(1) protocol and all processor $i$ in it, $i$ cannot compute $F$ correctly with probability better than $0.99$, over a uniform random matrix from $\{0,1\}^{n \times n}$.
\end{reminder}
\begin{proof}
	\newcommand{\eventbad}{\event_{\textsf{bad}}}
	\newcommand{\eventacc}{\event_{\textsf{acc}}}
	\newcommand{\Findep}{F_{\textsf{indep}}}
	\newcommand{\acc}{\textsf{acc}}
	
	Let $M$ be the input matrix from $\distr_{n \times n}$, where processor $i$ gets its $i$-th row. Let $\distr_A$ be the uniform distribution $\distr_{n \times n}$, and $\distr_B$ be the input distribution of case (B) in Theorem~\ref{theo:multi-round} when setting $k=n-1$. 
	
	We need some results about random $\mathbb{F}_2$ matrix from Section 3.2 of~\cite{kolchin1999random}. In particular, let $P_{n,s}$ be the probability that a uniformly random $\mathbb{F}_2$ matrix from $\mathbb{F}_2^{n\times n}$ has rank $n-s$. For all $s$, we have
	\[
	\lim_{n \to \infty} P_{n,s} = Q_s := 2^{-s^2} \cdot \left( \prod_{i \ge s+1} (1- 2^{-i}) \right) \cdot \left( \prod_{1\le i \le s} (1-2^{-i})^{-1} \right).
	\]
	
	Numerically, we have $Q_0 \approx 0.2887880950866$. Let $i$ be a processor, and $\acc(M)$ be $i$'s output on input matrix $M$, it suffices to show that $\acc(M)$ can not be correct w.r.t. $\Ffullrank$ with probability  higher than $0.99$. Set $\epsilon = 1 - 0.99 = 0.01$ for convenience.
	
	Suppose for the contradiction that $\acc(M)$ is correct w.r.t. $\Ffullrank(M)$ with probability at least $1-\epsilon$ over $M \sim \distr_A$, then we have
	\[
	\left| \Ex_{M \sim \distr_A}[\acc(M)] - \Ex_{M \sim \distr_A}[\Ffullrank(M)] \right| \le \epsilon,
	\]
	which means
	\[	
	\left| \Ex_{M \sim \distr_A}[\acc(M)] - Q_0 \right| \le \epsilon + o(1).
	\]	
	Also, by Theorem~\ref{theo:multi-round}, we have
	\[
	 \left| \Ex_{M \sim \distr_A}[\acc(M)] - \Ex_{M \sim \distr_B}[\acc(M)] \right| = o(1),
	\]
	and therefore
	\[	
	\left| \Ex_{M \sim \distr_B}[\acc(M)] - Q_0 \right| \le \epsilon + o(1).
	\]
	
	However, for all matrices in the support of $\distr_B$, their rank is at most $n-1$, which means $\acc(M)$ must be wrong on most of them. Note that	
	\[
	\Pr_{M \sim \distr_A}[\acc(M) \ne \Ffullrank(M)] \ge \Ex_{M \sim \distr_B}\left[ [\acc(M) \ne \Ffullrank(M)] \cdot  \frac{\distr_A(M)}{\distr_B(M)} \right],
	\]
	where $\distr_A(M)$ and $\distr_B(M)$ denote the probability of getting $M$ for distributions $\distr_A$ and $\distr_B$.
	
	For a matrix $M \sim \distr_B$, suppose the rank of its first $n-1$ columns is $n-s$, $\distr_B(M)$ can be computed as
	\[
	\distr_B(M) = 2^{-n(n-1)} \cdot 2^{-(n-s)} = 2^{-n^2} \cdot 2^{s}.
	\]
	
	Furthermore, for $M \sim \distr_B$, the probability that its first $n-1$ columns have rank at least $n-s$ is at least $\sum_{j=0}^{s-1} P_{n-1,j}$, as the rank of an $n \times (n-1)$ matrix is always no less than the rank of its left-top $(n-1) \times (n-1)$ matrix. Setting $s = 3$, we can see that for large enough $n$, with probability at least $\sum_{j=0}^{2} Q_j \ge 1 - 0.006$, the first $n-1$ columns of $M$ have rank at least $n-3$. In that case, $ \frac{\distr_A(M)}{\distr_B(M)} \ge 2^{-s} = 1/8$. 
	
	Putting them together, we have
	\begin{align*}
	    &\Pr_{M \sim \distr_A}[\acc(M) \ne \Ffullrank(M)] \\
	\ge &\Ex_{M \sim \distr_B}\left[ [\acc(M) \ne \Ffullrank(M)] \cdot  \frac{\distr_A(M)}{\distr_B(M)} \right]\\
	\ge &\Ex_{M \sim \distr_B}\Big[ [\acc(M) \ne \Ffullrank(M)] \cdot  [\text{the first $n-1$ columns of $M$ have rank at least $n-3$}] \Big] \cdot \frac{1}{8}\\
	\ge & \left( 1 - Q_0 -\epsilon - o(1) - 1 - 0.006 \right) \cdot \frac{1}{8} > 0.05,\\
	\end{align*}	
	contradiction, which completes the proof.
\end{proof}

Considering the problem that checking whether the top $k\times k$ sub-matrix has full-rank, one immediately get the following average case time-hierarchy theorem for $\BCAST(1)$.

\begin{reminder}{Theorem~\ref{theo:time-H}}
	For any $\omega(\log n) \le k \le n$, there is a function $F$ such that a $k$-round $\BCAST(1)$ protocol can compute exactly, while any $k/20$-round $\BCAST(1)$ protocols cannot compute $F$ correctly with probability $0.99$ over the uniform distribution.
\end{reminder}

\subsection{Proof of Theorem~\ref{theo:multi-round}}

We need the following technical lemma first, whose proof is deferred to the end of the section.

\begin{lemma}\label{lm:fourier2}
	Given a function $f : \{0,1\}^{k+1} \to \{0,1\}$ and a set $D \subseteq \{0,1\}^{k+1}$ with $|D| \ge 2^{k/2}$, let $\distr_{[b],D}$ and $\distr_{k+1,D}$ be the conditional distributions of $\distr_{[b]}$ and $\distr_{k+1}$ on the set $D$.\footnote{When $\distr_{[b]}$ ($\distr_{k+1}$) has no mass on $D$, we set $\distr_{[b],D}$ ($\distr_{k+1,D}$) to be the uniform distribution on $D$.} We have
	
	%For $1 - 2^{-k/9}$ fractions of $b \in \{0,1\}^{k}$, we have
	\[
	\Ex_{b \sim \distr_k}\| f(\distr_{[b],D}) - f(\distr_{k+1,D}) \| \le 2^{-k/9}.
	\]
\end{lemma}

\renewcommand{\FCon}{F}

Now we are ready to prove Theorem~\ref{theo:multi-round}.

\begin{proofof}{Theorem~\ref{theo:multi-round}}	
	Similar to the proof of Theorem~\ref{theo:one-round}, we will consider a slightly stronger model where we have $j \cdot n$ turns, and on the $t^{th}$ turn, processor $(t-1) \bmod{n} + 1$ gets to send a single bit. Recall that we use $\Prand^{(t)}$ to denote the distribution of the transcript of the first $t$ rounds when all processors get random input from $\distr_{k+1}$, and $\distrP_{[b]}^{(t)}$ to denote the distribution of the transcript of the first $t$ rounds when all processors get random input from $\distr_{[b]}$.
	
	Note that to prove the theorem, it suffices to show that the distribution $\Prand^{(j \cdot n)}$ is close to $\distrP_{[b]}^{(j \cdot n)}$ for most choices of $b$. For this purpose, we are going to prove the following inequality holds for any $t \le j \cdot n$:
	
	\begin{equation}\label{eq:ex-distance-small-multi}
	\Ex_{b \sim \distr_k} \left[ \| \Prand^{(t)} - \distrP_{[b]}^{(t)} \| \right] \le 2 \cdot t \cdot 2^{-k/9}.
	\end{equation}
	
	It is easy to see that plugging in $t = j \cdot n$, \eqref{eq:ex-distance-small-multi} implies the theorem. To prove \eqref{eq:ex-distance-small-multi} for all $t$, we induct on $t$. Clearly, \eqref{eq:ex-distance-small} holds when $t = 0$. So it suffices to show that when it holds for $t-1$, it also holds for $t$. Let $ i = (t-1) \pmod{n} + 1$ be the processor broadcasting at the $t$-th turn. 
	
	Again, for $b \in \{0,1\}^{k}$, we wish to bound $\| \Prand^{(t)} - \distrP_{[b]}^{(t)} \|$. By Lemma~\ref{lm:dist}, we have
	\begin{equation}
	\| \Prand^{(t)} - \distrP_{[b]}^{(t)} \| \le \| \Prand^{(t-1)} - \distrP_{[b]}^{(t-1)} \| + \Ex_{p \sim \Prand^{(t-1)}} \left[ \left\|f_i^{|p}(\distr_{k+1,p}) - f_i^{|p}(\distr_{[b],p})  \right\| \right].
	\label{eq:next-step-dist}
	\end{equation}
	
	In the above, $\distr_{k+1,p}$ and $\distr_{[b],p}$ are the distributions $\distr_{k+1}$ and $\distr_{[b]}$ conditioned on the previous transcript $p$. Specifically, let $D^{(t-1)}_p$ denote the set of inputs to $f_i$ which are consistent with the transcript $p$\footnote{That is, simulating $f_i$ with transcript $p$ on that input results in transcript $p$ itself.}, then $\distr_{[b],p}$ and $\distr_{k+1,p}$ are the conditional distributions of $\distr_{[b]}$ and $\distr_{k+1}$ on set $D^{(t-1)}_p$.
	
	Here, we wish to use Lemma~\ref{lm:fourier2} to bound the second term on the right side of \eqref{eq:next-step-dist}. To satisfy the requirement of Lemma~\ref{lm:fourier2}, we have to show that $D^{(t-1)}_p$ is a large subset of $\{0,1\}^{k+1}$ with high probability. Hence, we need the following claim, whose proof is deferred until we prove the theorem first.
	
	\begin{claim}\label{claim:large-D-p}
		For $t \le j \cdot n \le \frac{k\cdot n}{10}$, with probability $1 - 2^{-k/4}$ over $p \sim \Prand^{(t-1)}$, we have $|D^{(t-1)}_p| \ge  2^{k/2}$.
	\end{claim}
	
	Now, for each $b$, we define a score $s_b$ as follows:
	\[
	s_b := \Ex_{p \sim \Prand^{(t-1)}} \left[ \left\|f_i^{|p}(\distr_{k+1,p}) - f_i^{|p}(\distr_{[b],p})  \right\| \right].
	\]
	
	For $p \in \{0,1\}^{t-1}$, we set $s_{b,p} := \left\|f_i^{|p}(\distr_{k+1,p}) - f_i^{|p}(\distr_{[b],p})  \right\|$. Then by Lemma~\ref{lm:fourier2} and Claim~\ref{claim:large-D-p}, when $p \sim \Prand^{(t-1)}$, with probability at least $1 - 2^{-k/4}$, we have $ |D^{(t-1)}_p| \ge  2^{k/2} $ and
	\[
	\Ex_{b \sim \distr_k} [s_{b,p}] \le 2^{-k/9}.
	\]
	
	Therefore, it follows
	\[
	\Ex_{b \sim \distr_k} [s_b] = \Ex_{p \sim \Prand^{(t-1)}} \Ex_{b \sim \distr_k} [s_{b,p}] \le 2^{-k/4} + 2^{-k/9} \le 2 \cdot 2^{-k/9}.
	\]
	
	Now we have
	
	\begin{align*}
	\Ex_{b \sim \distr_k} \left[ \| \Prand^{(t)} - \distrP_{[b]}^{(t)} \| \right] &\le 
	\Ex_{b \sim \distr_k} \left[ \| \Prand^{(t-1)} - \distrP_{[b]}^{(t-1)} \| + \Ex_{p \sim \Prand^{(t-1)}} \left[ \left\|f_i^{|p}(\distr_{k+1,p}) - f_i^{|p}(\distr_{[b],p})  \right\| \right] \right]\\
	\le&2 \cdot (t-1) \cdot 2^{-k/9} + \Ex_{b \sim \distr_k} [s_b]\\
	\le&2 \cdot t \cdot 2^{-k/9}.
	\end{align*}
	
	The above proves \eqref{eq:ex-distance-small-multi} for $t$, which completes the whole proof.
	
\end{proofof}

Now we prove Claim~\ref{claim:large-D-p}.

\begin{proofof}{Claim~\ref{claim:large-D-p}}
	
	Let $t_1,t_2,\dotsc,t_{\ell}$ be the indices of all previous $\ell$ turns with processor $i$ broadcasting, before the current $t$-th turn. We have $\ell \le j \le k/10$.
	Let $x \in \{0,1\}^{k+1}$, note that $x$ is consistent with transcript $p$, if for all $a \in [\ell]$, we have
	\[
	f_i^{|p^{(t_a - 1)}}(x) = p_{t_a},
	\]
	where $p^{(t_a - 1)}$ denotes the first $t_a - 1$ bits of $p$. We set $\FCon_i(x,p) = 1$ if $x$ and $p$ are consistent, and $0$ otherwise.
	
	Consider the random process of generating $p \sim \Prand^{(t-1)}$, suppose inputs to all processors other than $i$ are fixed, let $x^{-i} = (x_1,x_2,\dotsc,x_{i-1},x_{i+1},\dotsc,x_{n}) \in \{0,1\}^{(n-1) \times (k+1)}$ be those fixed input. Let $P_{x^{-i}}^{(t)}$ be the distribution of the transcript when $x_i \sim \distr_{k+1}$, and all other processors get (fixed) input according to $x^{-i}$. 
	
	For a fixed $x^{-i}$, note that there are only $2^\ell$ possible transcripts $p$ from $P_{x^{-i}}^{(t-1)}$, as the transcript is determined after fixing the output of processor $i$ at all $\ell$ rounds.	Therefore, let $T(x^{-i},x_i)$ be the transcript when all processors get inputs according to $x^{-i}$ and $x_i$, we can see when $p \sim P_{x^{-i}}^{(t-1)}$, $\FCon_i(x,p) = 1$ if and only if $T(x^{-i},x) = p$. That is,
	\[
	P_{x^{-i}}^{(t-1)}(p) = \Pr_{x_i \sim \distr_{k+1}}[T(x^{-i},x_i) = p] = D_p^{(t-1)} / 2^{k+1}.
	\]
	
	In above $P_{x^{-i}}^{(t-1)}(p)$ is the probability that getting $p$ from distribution $P_{x^{-i}}^{(t-1)}$. Then we have
	\begin{align*}
	&\Pr_{p \sim P_{x^{-i}}^{(t-1)}} \left[ D_p^{(t-1)} < 2^{-\ell - k/4} \cdot 2^{k+1}  \right] \\
   =&\Pr_{p \sim P_{x^{-i}}^{(t-1)}} \left[ P_{x^{-i}}^{(t-1)}(p) < 2^{-\ell - k/4}  \right] \\
   \le&2^{-\ell - k/4} \cdot 2^{\ell} = 2^{-k/4}.
	\end{align*}
	The last inequality holds because the support size of $P_{x^{-i}}^{(t-1)}$ is at most $2^{\ell}$.
	
	Hence, we have
	\begin{align*}
	\Pr_{p \sim \Prand^{(t-1)}} \left[D_p^{(t-1)} < 2^{-\ell - k/4} \cdot 2^{k+1} \right] &= \Ex_{ x^{-i} \sim \distr_{(n-1) \times (k+1)} } \left[ \Pr_{p \sim P_{x^{-i}}^{(t-1)}} \left[D_p^{(t-1)} < 2^{-\ell - k/4} \cdot 2^{k+1} \right] \right]\\
	&\le 2^{-k/4}.
	\end{align*}
	
	The claim follows from that $\ell \le k/10$.	
\end{proofof}

\subsection{Proof of Lemma~\ref{lm:fourier2}}

Here we prove Lemma~\ref{lm:fourier2}.

\begin{proofof}{Lemma~\ref{lm:fourier2}}
	Let $g$ be the following function
	\[
	g(x) := \begin{cases}
	f(x) \quad &x \in D\\
	0 \quad &x \notin D.
	\end{cases}
	\]
	Let $N_D := |D|$, and for $b \in \{0,1\}^{k}$, let $D_{[b]}$ be the support set of $\distr_{[b]}$ and $N_b := |D \cap D_{[b]}|$. That is, $N_D$ is the size of the support set of $\distr_{k+1,D}$, while $N_b$ is the size of the support size of $\distr_{b,D}$.
	
	We need the following claim, which shows that for most $b$'s $N_b$ is close of a half of $N_D$. We defer its proof until we prove the lemma. 
	
	\begin{claim}~\label{claim:close}
		Let $b \sim \distr_{k}$, with probability $1 - 2^{-k/8}$, we have $| N_b/N_D - \frac{1}{2}| < 2^{-k/8}$.
	\end{claim}

	By Lemma~\ref{lm:fourier}, we have
	\[
	\sum_{b \in \{0,1\}^{k}} \| g(\distr_{[b]}) - g(\distr_{k+1}) \|^2 \le \Ex_{x \sim \distr_{k+1}}[g(x)] \le N_D / 2^{k+1} \le N_D / 2^{k}.
	\]
	Equivalently,
	\begin{equation}
	\Ex_{b \sim \distr_{k}} \left[ \| g(\distr_{[b]}) - g(\distr_{k+1}) \|^2 \right] = 
	\Ex_{b \sim \distr_{k}} \left[ \left| \Ex_{x\sim\distr_{[b]}}[g(x)] - \Ex_{x\sim\distr_{k+1}}[g(x)] \right|^2 \right]  \le N_D / 2^{2 k}.\label{eq:bound-for-g}
	\end{equation}
	
	In order to make use of the above bound \eqref{eq:bound-for-g}, we now relate $g(\distr_{[b]})$ and $g(\distr_{k+1})$ to $f(\distr_{[b],D})$ and $f(\distr_{k+1,D})$. From the definition of $g$, we have
	\begin{align}
	\Ex_{x \sim \distr_{[b],D}}[f(x)] &= \Ex_{x\sim\distr_{[b]}}[g(x)] \cdot \frac{2^{k}}{N_b}\notag\\
									  &= \Ex_{x\sim\distr_{[b]}}[g(x)] \cdot \frac{2^{k}}{N_D / 2} \cdot \frac{N_D}{2 N_b},\label{eq:lst-line}
	\end{align}
	and
	\begin{align*}
	\Ex_{x \sim \distr_{k+1,D}}[f(x)] &= \Ex_{x\sim\distr_{k+1}}[g(x)] \cdot \frac{2^{k+1}}{N_D} \\
									  &= \Ex_{x\sim\distr_{k+1}}[g(x)] \cdot \frac{2^{k}}{N_D/2}.
	\end{align*}
	
	That is, $\Ex_{x \sim \distr_{[b],D}}[f(x)]$ and $\Ex_{x \sim \distr_{k+1,D}}[f(x)]$ are $ \Ex_{x\sim\distr_{[b]}}[g(x)]$ and $\Ex_{x\sim\distr_{k+1}}[g(x)]$ scaled by a factor of $\frac{2^k}{N_D/2}$, except for another $\frac{N_D}{2 N_b}$ factor in \eqref{eq:lst-line}, which is very close to $1$ for most $b$'s by Claim~\ref{claim:close}.
	
	We first ignore the $\frac{N_D}{2 N_b}$ factor in \eqref{eq:lst-line}, and define
	\[
	F_{[b]} := \Ex_{x\sim\distr_{[b]}}[g(x)] \cdot \frac{2^{k}}{N_D/2},
	\]
	and
	\[
	F_{k+1} := \Ex_{x \sim \distr_{k+1,D}}[f(x)] = \Ex_{x\sim\distr_{k+1}}[g(x)] \cdot \frac{2^k}{N_D / 2}.
	\]
	
	Scaling each side of \eqref{eq:bound-for-g} by $\left( \frac{2^k}{N_D / 2} \right)^2$, we have
	\[
	\Ex_{b \sim \distr_{k}} \left[ | F_{[b]} - F_{k+1} |^2 \right] \le  N_D / 2^{2 k} \cdot \left( \frac{2^k}{N_D / 2} \right)^2 = 4 / N_D \le 4 \cdot 2^{-k/2},
	\]
	the last inequality follows from the assumption that $N_D = |D| \ge 2^{-k/2}$.
	
	That is, by Markov's inequality, when $b \sim \distr_k$, with probability $1-2^{-k/8}$, we have 
	\[
	| F_{[b]} - F_{k+1} |^2 \le 4 \cdot 2^{-k/2} \cdot 2^{k/8} < 2^{-k/4},
	\]
	which means $|F_{[b]} - F_{k+1}| < 2^{-k/8}$.
	
	Now we take care of the additional $\frac{N_D}{2 N_b}$ factor in \eqref{eq:lst-line}. By Claim~\ref{claim:close}, when $b \sim \distr_k$, with probability $1 - 2^{-k/8}$, we have $\left|N_b / N_D - \frac{1}{2}\right| < 2^{-k/8}$, which means $\left|\frac{N_D}{2 N_b} - 1\right| < 3 \cdot 2^{-k/8}$.
	
	Putting everything together, when $b \sim \distr_k$, with probability $1 - 2 \cdot 2^{-k/8}$, we have
	\begin{align*}
	\| f(\distr_{[b],D}) - f(\distr_{k+1,D}) \| &= \left| \Ex_{x \sim \distr_{[b],D}}[f(x)] - \Ex_{x \sim \distr_{k+1,D}}[f(x)] \right|\\
												&= \left| F_{[b]} \cdot \frac{N_D}{2 N_b} - F_{k+1} \right|\\
												&\le \left| F_{[b]} - F_{k+1} \right| + F_{[b]} \cdot \left|\frac{N_D}{2 N_b} - 1\right|\\
												&\le 4 \cdot 2^{-k/8}.
	\end{align*}
	
	The last line follows from that $F_{[b]} = \Ex_{x\sim\distr_{[b]}}[g(x)] \cdot \frac{2^{k+1}}{N_D} \le \frac{N_D}{2^{k+1}} \cdot \frac{2^{k+1}}{N_D} = 1$.
	
	Therefore,
	\[
	\Ex_{b \sim \distr_k}\| f(\distr_{[b],D}) - f(\distr_{k+1,D}) \| \le (2 \cdot 2^{-k/8}) + 4 \cdot 2^{-k/8} \le 2^{-k/9}.
	\]

\end{proofof}

Finally, we prove Claim~\ref{claim:close}.

\begin{proofof}{Claim~\ref{claim:close}}
	Let $I$ be the indicator function for set $D$:
	\[
	I(x) := \begin{cases}
	1 \quad &x \in D\\
	0 \quad &x \notin D.
	\end{cases}
	\]
	
	By Lemma~\ref{lm:fourier}, we have
	\begin{equation}
	\sum_{b \in \{0,1\}^{k}} \left| \Ex_{x \sim \distr_{[b]}}[I(x)] - \Ex_{x \sim \distr_{k+1}}[I(x)] \right|^2 \le \Ex_{x \sim \distr_{k+1}}[I(x)] = N_D / 2^{k+1} \le N_D / 2^{k}.
	\label{eq:bound-for-I}
	\end{equation}
	
	Note that from the definitions, $\Ex_{x \sim \distr_{[b]}}[I(x)] = N_b / 2^k$ and $\Ex_{x \sim \distr_{k+1}}[I(x)] = N_D / 2^{k+1}$. Plugging these in~\eqref{eq:bound-for-I}, we have
	\[
	\sum_{b \in \{0,1\}^{k}} \left| N_b / 2^{k} - N_D / 2^{k+1} \right|^2 \le N_D / 2^{k}.
	\]
	Scaling each side by $\left(2^{k} / N_D\right)^2$, we have
	\[
	\sum_{b \in \{0,1\}^{k}} \left| N_b / N_D -\frac{1}{2} \right|^2 \le N_D / 2^{k} \cdot \left(2^{k} / N_D\right)^2 = 2^{k} / N_D.
	\]
	Equivalently,
	\[
	\Ex_{b \sim \distr_k} \left[ \left| N_b / N_D -\frac{1}{2} \right|^2 \right] \le 1 / N_D \le 2^{-k/2},
	\]	
	where the last inequality follows from the assumption that $N_D = |D| \ge 2^{k/2}$. Finally, by Markov's inequality, when $b \sim \distr_k$, with probability $1 - 2^{-k/8}$, we have $\left| N_b / N_D -\frac{1}{2} \right|^2 \le 2^{-k/2} \cdot 2^{k/8} \le 2^{-k/4}$, and it follows $\left| N_b / N_D -\frac{1}{2} \right| < 2^{-k/8}$, which completes the proof.
\end{proofof}
\section{The Complete Pseudo-random Generator }\label{sec:completePRG}

In this section we construct the PRG.

\begin{reminder}{Theorem~\ref{theo:PRG-formal}}
	For all $m = O(n)$ and $k = \Omega(\log n)$, there exists an $(O(k), m, n, \Omega(k))$ \BCAST(1) PRG that can be constructed within $O(k)$ rounds. In particular, the PRG works as follows
	
	\begin{itemize}
		\item Each processor gets $k + k \cdot \frac{(m-k)}{n}$ private random bits.
		
		\item Then in $O\left(\frac{m-k}{n} \cdot k \right) = O(k)$ rounds, all processors broadcast their last $k \cdot \frac{(m-k)}{n}$ random bits. And they use that to construct a random matrix $M \in \{0,1\}^{k \times (m-k)}$.
		
		\item Each processor's output is simply the concatenation of its first $k$ random bits $x$ and $x^{T} M$.
	\end{itemize}
\end{reminder}
%The above theorem follows immediately from Theorem \ref{theo:PRG} (restated below) by setting $m = n$. We spend the remainder of this section proving Theorem \ref{theo:PRG}.

The following corollary follows directly from the above theorem.
\begin{cor}\label{cor:save-random-bits}
	Let $A$ be a $k$-round randomized \BCAST(1) algorithm with poly($n$) time processors, where each processor uses up to $n$ random bits and $k = \Omega(\log n)$. Then there exists an algorithm $A'$ solving the same problem within $O(k)$-rounds, where each processor uses at most $k$ random bits.
\end{cor}

Note that the correctness of Theorem~\ref{theo:PRG-formal} follows directly from Theorem~\ref{theo:PRG} (restated below). We spend the remainder of this section proving Theorem \ref{theo:PRG}.

\paragraph*{Notations.} We first recall some notations. Let $M \in \{0,1\}^{n \times m}$. We use $\distr_{M}$ to denote the uniform distribution on the following set $\{ (x,x^{T} M)  : x \in \{0,1\}^{n} \}$, which is a subset of $\{0,1\}^{n+m}$. For integers $n$ and $m$, we use $\distr_{n \times m}$ to denote the uniform distribution on $\{0,1\}^{n \times m}$.

Supposing there are $n$ processors in total, we are going to assume they are deterministic. Processor $i$ can be defined by a function $f_i : \{0,1\}^{m} \times \{0,1\}^{*} \to \{0,1\}$, such that $f_i(z,p)$ is the bit player $i$ outputs when it gets the input $z$ and previous history $p$. We are going to use $f_i^{|p}$ to denote the function $f_i(\cdot,p)$ for simplicity. If transcript $p$ is incompatible with player $i$ having input $z$, then we set $f_i(z, p)$ arbitrarily.

\begin{reminder}{Theorem~\ref{theo:PRG}}
	Let $n,m,k$ be three integers. Consider the following two cases:
	
	\begin{itemize}
		\item (A) All processors receive random inputs from $\distr_{m}$.
		\item (B) Let $M$ be a uniform sample from $\distr_{k \times (m-k)}$, then all processors receive inputs from $\distr_{M}$.
	\end{itemize}
	
	%Let $U'$ be the distribution $\sum_{b \in \mathbb{F}_2^k} \frac{1}{2^k} U_{[b]}$, where $U_{[b]}$ is the distribution where processor $i$ receives a random vector $a_i \in \mathbb{F}_2^k$, along with $a_i \cdot b$.
	
	For $j \le k/10$, $m \le 2^{k/20}$ and any $j$-round \BCAST(1) protocol, the statistical distance between the distributions of its transcripts in case (A) and (B) is at most $O\left( \frac{j \cdot n}{2^{k/9}} \right)$.
\end{reminder}

To prove Theorem~\ref{theo:PRG}, we need the following technical lemma, whose proof is deferred to the end of this section.

\begin{lemma}\label{lm:fourier4}
Assuming $m \le 2^{k/20}$, given a function $f : \{0,1\}^{m} \to \{0,1\}$ and a set $D \subseteq \{0,1\}^{m}$ with $|D| \ge 2^{m-k/2}$, let $\distr_{M,D}$ and $\distr_{m,D}$ be the conditional distributions of $\distr_{M}$ and $\distr_{m}$ on the set $D$.\footnote{When $\distr_{M}$ ($\distr_{m}$) has no mass on $D$, we set $\distr_{M,D}$ ($\distr_{m,D}$ to be the uniform distribution on $D$.)} We have
\[
\Ex_{M \sim \distr_{k\times(m-k)}} \| f(\distr_{M,D}) - f(\distr_{m,D}) \| \le 2^{-k/9}.
\]
\end{lemma}

\begin{proofof}{Theorem~\ref{theo:PRG}}	
	Similar to the proof of Theorem~\ref{theo:one-round}, we will consider a slightly stronger model where we have $j \cdot n$ turns, and on the $t^{th}$ turn, processor $(t-1) \bmod{n} + 1$ gets to send a single bit. We use similar notations as in the proof of Theorem~\ref{theo:one-round} and Theorem~\ref{theo:multi-round}. Let $\Prand^{(t)}$ be the distribution of the transcripts of the first $t$ rounds when all processors get random input from $\distr_{m}$. For a matrix $M \in \{0,1\}^{k \times (m-k)}$, we use $P_{M}^{(t)}$ to denote the distribution of the transcript of the first $t$ rounds when all processors get random input from $\distr_{M}$. Recall that $\distr_M$ is the uniform distribution on the set $\{ (x,x^{T} M)  : x \in \{0,1\}^{k} \}$.
	
	Note that to prove the theorem, it suffices to show that the distribution $\Prand^{(j \cdot n)}$ is close to $\distrP_{M}^{(j \cdot n)}$ for most choices of $M \sim \distr_{k \times (m - k)}$. For this purpose, we are going to prove the following inequality holds for any $t \le j \cdot n$:
	
	\begin{equation}\label{eq:ex-distance-small-full}
	\Ex_{M \sim \distr_{k\times (m-k)}} \left[ \| \Prand^{(t)} - \distrP_{M}^{(t)} \| \right] \le 2 \cdot t \cdot 2^{-k/9}.
	\end{equation}
	
	It is easy to see that plugging in $t = j \cdot n$, \eqref{eq:ex-distance-small-full} implies the theorem. To prove \eqref{eq:ex-distance-small-full} for all $t$, we induct on $t$. Clearly, \eqref{eq:ex-distance-small} holds when $t = 0$. So it suffices to show that when it holds for $t-1$, it also holds for $t$. Let $ i = (t-1) \pmod{n} + 1$ be the processor broadcasting at the $t$-th turn.
	
	For an $M \in \{0,1\}^{k \times (m-k)}$, we wish to bound $\| \Prand^{(t)} -\distrP_{M}^{(t)} \|$. By Lemma~\ref{lm:dist}, we have
	\begin{equation}
	\| \Prand^{(t)} -\distrP_{M}^{(t)} \| \le \| \Prand^{(t-1)} -\distrP_{M}^{(t-1)} \| + \Ex_{p \sim \Prand^{(t-1)}} \left[ \left\|f_i^{|p}(\distr_{m,p}) - f_i^{|p}(\distr_{M,p})  \right\| \right].
	\label{eq:next-step-dist-2}
	\end{equation}
	
	In above, $\distr_{m,p}$ and $\distr_{M,p}$ are distributions $\distr_{m}$ and $\distr_{M}$ conditioned on the previous transcript $p$. Specifically, let $D^{(t-1)}_p$ denote the set of inputs to $f_i$ which are consistent with the transcript $p$, then $\distr_{m,p}$ and $\distr_{M,p}$ are the conditional distributions of $\distr_{m}$ and $\distr_{M}$ on set $D^{(t-1)}_p$.
	
	We wish to use Lemma~\ref{lm:fourier4} to bound the second term on the right side of \eqref{eq:next-step-dist-2}. To do so, we need to show $D^{(t-1)}_p$ is large with high probability for $p \sim \Prand^{(t-1)}$. The following claim can be proved in exactly the same way as Claim~\ref{claim:large-D-p} in the proof of Theorem~\ref{theo:multi-round}.
    
    \begin{claim}\label{claim:big-D-p-2}
    For $t \le j \cdot n \le \frac{k\cdot n}{10}$, with probability $1 - 2^{-k/4}$ over $p \sim \Prand^{(t-1)}$, we have $|D^{(t-1)}_p| \ge  2^{m-k/4}$.
    \end{claim}
    
    Now, for each $M \in \{0,1\}^{k \times (m-k)}$, we again define a score $s_M$ as follows:
    \[
    s_M := \Ex_{p \sim \Prand^{(t-1)}} \left[ \left\|f_i^{|p}(\distr_{m,p}) - f_i^{|p}(\distr_{M,p})  \right\| \right].
    \]
    
    For $p \in \{0,1\}^{t-1}$, we set $s_{M,p} := \left\|f_i^{|p}(\distr_{m,p}) - f_i^{|p}(\distr_{M,p})  \right\|$. Then by Lemma~\ref{lm:fourier4} and Claim~\ref{claim:big-D-p-2}, when $p \sim \Prand^{(t-1)}$, with probability at least $1 - 2^{-k/4}$, we have $|D^{(t-1)}_p| \ge  2^{m-k/4}$, and therefore
    
    \[
    \Ex_{M \sim \distr_{k \times (m-k)}} \left[ s_{M,p} \right] \le 2^{-k/9}.
    \]
    
    Therefore, it follows
    \[
    \Ex_{M \sim \distr_{k \times (m-k)}} [s_M] = \Ex_{p \sim \Prand^{(t-1)}} \Ex_{M \sim \distr_{k \times (m-k)}} [s_{M,p}] \le 2^{-k/4} + 2^{-k/9} \le 2 \cdot 2^{-k/9}.
    \]
    
    Now we have
    
    \begin{align*}
    \Ex_{M \sim \distr_{k \times (m-k)}} \left[ \| \Prand^{(t)} - \distrP_{[b]}^{(t)} \| \right] &\le 
    \Ex_{M \sim \distr_{k \times (m-k)}} \left[ \| \Prand^{(t-1)} - \distrP_{[b]}^{(t-1)} \| + \Ex_{p \sim \Prand^{(t-1)}} \left[ \left\|f_i^{|p}(\distr_{m,p}) - f_i^{|p}(\distr_{M,p})  \right\| \right] \right]\\
    \le&2 \cdot (t-1) \cdot 2^{-k/9} + \Ex_{M \sim \distr_{k \times (m-k)}} [s_M]\\
    \le&2 \cdot t \cdot 2^{-k/9}.
    \end{align*}
    
    The above proves \eqref{eq:ex-distance-small-full} for $t$, which completes the whole proof.
\end{proofof}

\subsection{Proof of Lemma~\ref{lm:fourier4}}

Before proving Lemma~\ref{lm:fourier4}, we first prove the following technical lemma, which is a generalization of Lemma~\ref{lm:fourier}.

\begin{lemma}\label{lm:fourier3}
	Given a function $f : \{0,1\}^{m} \to \{0,1\}$, we have
    \[
    \Ex_{M \sim \distr_{k \times (m-k)}} \left[ \left\| f(\distr_{m}) - f(\distr_{M}) \right\|^2 \right] \le 2^{-k} \cdot (m-k)^2 \cdot \Ex_{x \sim \distr_{m}}[f(x)].
    \]
\end{lemma}
\begin{proof}
	Let $M \sim \distr_{k \times (m - k)}$. Let $v_1,v_2,\dotsc,v_{m-k}$ be all the $m - k$ columns of $M$, such that $v_1$ is the last column and $v_{m-k}$ is the first. Clearly, $u_i$'s are i.i.d. samples from $\distr_{k}$.
	
	We are going to prove this lemma via a hybrid argument, for a fixed $M$, let $\distr_{M,j}$ be the uniform distribution on the following set
	\[
	\{ (x,x^{(k)} \cdot v_{j},x^{(k)} \cdot v_{j-1},\dotsc,x^{(k)} \cdot v_1) : x \in \{0,1\}^{m-j} \},
	\]
	where $x^{(k)}$ denotes the first $k$ bits of string $x$. That is, for $x \sim \distr_{M,j}$, the first $m - j$ bits are completely random, while the last $j$ bits are generated according to $M$.	By definition, it is easy to see that $\distr_{M,0} = \distr_{m}$, and $\distr_{M,m-k} = \distr_{M}$. 
	
	The following claim is the central ingredient of our hybrid argument. 
	\begin{claim}\label{claim:hybrid}
		For $0 \le j < m - k$, we have
		\[
		\Ex_{M \sim \distr_{k \times (m-k)}} \| f(\distr_{M,j}) - f(\distr_{M,j+1}) \|^2 \le 2^{-k} \cdot \Ex_{x \sim \distr_{m}}[f(x)].
		\]
	\end{claim}
	
	Before proving Claim~\ref{claim:hybrid}, we show it implies our lemma.
	
	First, for $k$ reals $a_1,a_2,\dotsc,a_k$, we have $\|a\|_1 \le \sqrt{k} \cdot \|a\|_2 $, and consequently
	\begin{equation}
	\left( \sum_{i=1}^{k} a_i \right)^2 \le k \cdot \sum_{i=1}^{k} a_i^2. \label{eq:ell-1-and-ell-2}
	\end{equation}
	
	So we have
	\begin{align*}
	 \Ex_{M \sim \distr_{k \times (m-k)}} \left\| f(\distr_{m}) - f(\distr_{M}) \right\|^2 
   &=\Ex_{M \sim \distr_{k \times (m-k)}} \left\| f(\distr_{M,0}) - f(\distr_{M,m-k}) \right\|^2\\
   &=\Ex_{M \sim \distr_{k \times (m-k)}} (m-k) \cdot \sum_{j=0}^{m-k-1} \left\| f(\distr_{M,j}) - f(\distr_{M,j+1}) \right\|^2 \tag{by \eqref{eq:ell-1-and-ell-2}}\\
   &= (m-k) \cdot \sum_{j=0}^{m-k-1} \cdot \Ex_{M \sim \distr_{k \times (m-k)}} \left\| f(\distr_{M,j}) - f(\distr_{M,j+1}) \right\|^2\\
   &\le 2^{-k} \cdot (m-k)^2 \cdot \Ex_{x \sim \distr_{m}}[f(x)] \tag{by Claim~\ref{claim:hybrid}}.
	\end{align*}

	Finally, we prove Claim~\ref{claim:hybrid}.
	\begin{proofof}{Claim~\ref{claim:hybrid}}\renewcommand{\qedsymbol}{}
		\newcommand{\firstvj}{{v^{(j)}}}
		First, note that $v_{j+2},\dotsc,v_{m-k}$ are not involved in the inequality in the claim. Suppose we fix $v_{1},v_{2},\dotsc,v_{j}$ first. We use $\firstvj$ to denote this vector sequence.
		
		Now, we define the extension function $E_{\firstvj} : \{0,1\}^{m-j} \to \{0,1\}^{m}$ as follows
		\[
		E_{\firstvj}(x) := (x,x^{(k)} \cdot v_j, x^{(k)} \cdot v_{j-1},\dotsc,x^{(k)} \cdot v_{1}).
		\]
		That is, extending the vector $x$ as if $\firstvj$ is the last $j$ columns of the matrix $M$.
		
		We also define $g_\firstvj : \{0,1\}^{m-j} \to \{0,1\}$ by composing $E_{\firstvj}$ and $f$:
		\[
		g_\firstvj(x) := f(E_\firstvj(x)).
		\]
		
		Let $k' = m -j - 1$, and $\distr_{[v_{j+1}]}$ be the uniform distribution on the set $ \{ (x,x^{(k)} \cdot v_{j+1}) : x \in \{0,1\}^{k'} \}$. By a similar proof of Lemma~\ref{lm:fourier}, we have
		\[
		\sum_{v_{j+1} \in \{0,1\}^{k}} \| g_\firstvj(\distr_{[v_{j+1}]}) - g_\firstvj(\distr_{k'+1}) \|^2 \le \Ex_{x \sim \distr_{k'+1}} [g_{\firstvj}(x)],
		\]
		or equivalently,
		\begin{equation}
		\Ex_{v_{j+1} \sim \distr_k} \| g_{\firstvj}(\distr_{[v_{j+1}]}) - g_{\firstvj}(\distr_{k'+1}) \|^2 \le 2^{-k} \cdot \Ex_{x \sim \distr_{k'+1}} [g_{\firstvj}(x)].\label{eq:for-one-fixed-firstvj}
		\end{equation}
		
		Averaging over all $v_{1},v_{2},\dotsc,v_{j}$ from $\distr_{k}$, from the definition of $\distr_{M,j}$ and $\distr_{M,j+1}$, the left side of~\eqref{eq:for-one-fixed-firstvj} becomes
		\[
		\Ex_{v_1,v_2,\dotsc,v_{j} \sim \distr_k} \left[ \Ex_{v_{j+1} \sim \distr_k} \| g_\firstvj(\distr_{[v_{j+1}]}) - g_\firstvj(\distr_{m-j}) \|^2 \right] = \Ex_{M \sim \distr_{k \times (m-k)}} \| f(\distr_{M,j}) - f(\distr_{M,j+1}) \|^2,
		\]
		and the right side becomes
		\[
		\Ex_{v_1,v_2,\dotsc,v_{j} \sim \distr_k} \left[ 2^{-k} \cdot \Ex_{x \sim \distr_{m-j}} [g_{\firstvj}(x)] \right] = 2^{-k} \cdot \Ex_{x \sim \distr_{m}}[f(x)],
		\]
		which completes the proof.
	\end{proofof}
\end{proof}

Now we move to prove Lemma~\ref{lm:fourier4}.

\begin{proofof}{Lemma~\ref{lm:fourier4}}
	The following proof are quite similar to the proof of Lemma~\ref{lm:fourier2}. Let $g$ be the following function
	\[
	g(x) := \begin{cases}
	f(x) \quad &x \in D\\
	0 \quad &x \notin D.
	\end{cases}
	\]
	Let $N_D := |D|$, and for $M \in \{0,1\}^{k \times (m-k)}$, let $D_{M}$ be the support set of $\distr_{[M]}$, and $N_M := |D \cap D_{M}|$. That is, $N_D$ is the support size of distribution $\distr_{m,D}$, while $N_M$ is the support size of $\distr_{M,D}$.
	
	We need the following claim, whose proof is deferred until we prove the lemma.
	
	\begin{claim}~\label{claim:close2}
		Let $M \sim \distr_{k \times (m-k)}$, with probability $1 - 2^{-k/8}$, we have
		\[
		\left| N_M / N_D - 2^{-(m-k)} \right| \le 2^{-k/8} \cdot 2^{-(m-k)}.
		\]
	\end{claim}
	
	By Lemma~\ref{lm:fourier}, we have
	\[
	\Ex_{M \sim \distr_{k \times (m-k)}} \| g(\distr_{M}) - g(\distr_{m}) \|^2 \le m^2 \cdot 2^{-k} \cdot \Ex_{x \sim \distr_{m}}[g(x)] \le m^2 \cdot 2^{-k} \cdot N_D / 2^{m}.
	\]
	Equivalently,
	\begin{equation}
	\Ex_{M \sim \distr_{k \times (m-k)}} \left[ \left| \Ex_{x\sim\distr_{M}}[g(x)] - \Ex_{x\sim\distr_{m}}[g(x)] \right|^2 \right] \le m^2 \cdot 2^{-k} \cdot N_D / 2^{m}.
	\label{eq:bound-for-g-2}
	\end{equation}
	
	To make use of the above bound \eqref{eq:bound-for-g-2}, we now relate $g(\distr_{M})$ and $g(\distr_{m})$ to $f(\distr_{M,D})$ and $f(\distr_{m,D})$. From the definition of $g$, we have
	\begin{align}
	\Ex_{x \sim \distr_{M,D}}[f(x)] &= \Ex_{x\sim\distr_{M}}[g(x)] \cdot \frac{2^{k}}{N_M} \tag{the support size of $\distr_M$ is $2^k$} \\
	                                &= \Ex_{x\sim\distr_{M}}[g(x)] \cdot \frac{2^{k}}{N_D / 2^{m-k}} \cdot \frac{N_D}{2^{m-k} N_M} \label{eq:lst-line-2},
	\end{align}
	and
	\begin{align*}
	\Ex_{x \sim \distr_{m,D}}[f(x)] &= \Ex_{x\sim\distr_{m}}[g(x)] \cdot \frac{2^m}{N_D} \\
	                                &= \Ex_{x\sim\distr_{m}}[g(x)] \cdot \frac{2^k}{N_D / 2^{m-k}}.
	\end{align*}
	
	That is, $\Ex_{x \sim \distr_{M,D}}[f(x)]$ and $\Ex_{x \sim \distr_{m,D}}[f(x)]$ are $ \Ex_{x\sim\distr_{M}}[g(x)]$ and $\Ex_{x\sim\distr_{m}}[g(x)]$ scaled by a factor of $\frac{2^k}{N_D / 2^{m-k}}$, except for another $\frac{N_D}{2^{m-k} N_M}$ factor in \eqref{eq:lst-line-2}, which is very close to $1$ for most $M$'s by Claim~\ref{claim:close2}.
	
	We first ignore the $\frac{2^k}{N_D / 2^{m-k}}$ factor in \eqref{eq:lst-line-2}, and define
	\[
	F_{M} := \Ex_{x\sim\distr_{M}}[g(x)] \cdot \frac{2^{k}}{N_D / 2^{m-k}},
	\]
	and
	\[
	F_{m} := \Ex_{x \sim \distr_{m,D}}[f(x)] = \Ex_{x\sim\distr_{m}}[g(x)] \cdot \frac{2^k}{N_D / 2^{m-k}}.
	\]
	
	Scaling each side of \eqref{eq:bound-for-g-2} by $\left( \frac{2^k}{N_D / 2^{m-k}} \right)^2$, we have
	\begin{align*}
	\Ex_{M \sim \distr_{k \times (m-k)}} \left[| F_{M} - F_{m} |^2\right] &\le  m^2 \cdot 2^{-k} \cdot N_D / 2^{m} \cdot \left( \frac{2^k}{N_D / 2^{m-k}} \right)^2 \\
															 &\le 2^{k/10} \cdot 2^{m-k} / N_D \tag{$m \le 2^{k/20}$}\\
															 &\le 2^{k/10} \cdot 2^{-k/2} \tag{$N_D = |D| \ge 2^{m-k/2}$}.
	\end{align*}
	
	That is, by Markov's inequality, when $M \sim \distr_{k \times (m-k)}$, with probability $1-2^{-k/8}$, we have 
	\[
	| F_{M} - F_{k+1} |^2 \le 2^{k/10} \cdot 2^{-k/2} \cdot 2^{k/8} \le 2^{-k/4},
	\]
	which means $|F_{M} - F_{k+1}| < 2^{-k/8}$.
	
	Now, by Claim~\ref{claim:close2}, when $M \sim \distr_{k \times (m-k)}$, with probability $1 - 2^{-k/8}$, we have 
	\begin{align*}
			   &|N_M / N_D - 2^{-(m-k)}| < 2^{-k/8} \cdot 2^{-(m-k)} \\
	\Rightarrow&\left|\frac{N_M \cdot 2^{m-k}}{N_D} - 1\right| < 2^{-k/8} \\
	\Rightarrow&\left|\frac{N_D}{N_M \cdot 2^{m-k}} - 1\right| < 2 \cdot 2^{-k/8}.\\
	\end{align*}
	
	Putting everything together, with probability $1 - 2 \cdot 2^{-k/8}$, we have
	\begin{align*}
	\| f(\distr_{M,D}) - f(\distr_{m,D}) \| &= \left| \Ex_{x \sim \distr_{M,D}}[f(x)] - \Ex_{x \sim \distr_{m,D}}[f(x)] \right|\\
	&= \left| F_{M} \cdot \frac{N_D}{2^{m-k} N_M} - F_{m} \right|\\
	&\le \left| F_{M} - F_{m} \right| + F_{M} \cdot \left|\frac{N_D}{2^{m-k} N_M} - 1\right|\\
	&\le 3 \cdot 2^{-k/8}.
	\end{align*}
	
	That last line follows from that $F_{M} = \Ex_{x\sim\distr_{M}}[g(x)] \cdot \frac{2^{k}}{N_D / 2^{m-k}} \le \frac{N_D}{2^{m}} \cdot \frac{2^m}{N_D} = 1$.
	
	Therefore, we have
	\[
	\Ex_{M \sim \distr_{k\times(m-k)}} \| f(\distr_{M,D}) - f(\distr_{m,D}) \| \le 2 \cdot 2^{-k/8} + 3 \cdot 2^{-k/8} \le 2^{-k/9}.
	\]
\end{proofof}

Finally, we prove Claim~\ref{claim:close2}.

\begin{proofof}{Claim~\ref{claim:close2}}
	Let $I$ be the indicator function for set $D$:
	\[
	I(x) := \begin{cases}
	1 \quad &x \in D\\
	0 \quad &x \notin D.
	\end{cases}
	\]
	
	By Lemma~\ref{lm:fourier3}, we have
	\begin{equation}
	\Ex_{M \sim \distr_{k \times (m-k)}} \left[\left| \Ex_{x \sim \distr_{M}}[I(x)] - \Ex_{x \sim \distr_{m}}[I(x)] \right|^2\right] \le 2^{-k} \cdot m^2 \cdot \Ex_{x \sim \distr_m}[I(x)]
	\label{eq:bound-for-I-2}
	\end{equation}
	
	Note that $\Ex_{x \sim \distr_{M}}[I(x)] = N_M / 2^k$ and $\Ex_{x \sim \distr_{m}}[I(x)] = N_D / 2^m$. Plugging these in \eqref{eq:bound-for-I-2}, we have
	\[
		\Ex_{M \sim \distr_{k \times (m-k)}} \left| N_M / 2^{k} - N_D / 2^{m} \right|^2 \le 2^{-k} \cdot m^2 \cdot N_D / 2^{m}.
	\]
	
	Scaling both sides by $\left(\frac{2^{k}}{N_D} \right)^2$, we have	
	\begin{align*}
	\Ex_{M \sim \distr_{k \times (m-k)}} \left[ \left| N_M / N_D - 2^{-(m-k)} \right|^2 \right] &\le 2^{-k} \cdot m^2 \cdot N_D / 2^{m} \cdot \left(\frac{2^{k}}{N_D} \right)^2\\
																				 &\le 2^{-(m-k)} \cdot m^2 / N_D\\
																				 &\le 2^{-(m-k)} \cdot m^2 \cdot 2^{-(m - k + k/2)} \tag{$N_D = |D| \ge 2^{m-k/2}$}\\
																				 &\le 2^{-2(m-k) - k/2} \cdot 2^{k/10} \tag{$m \le 2^{k/20}$}.\\
	\end{align*}
	
	By Markov's inequality, when $M \sim \distr_{k \times (m-k)}$, with probability $1 - 2^{-k/8}$, we have 
	\[
	\left| N_M / N_D - 2^{-(m-k)} \right|^2 \le 2^{-2(m-k) - k/2} \cdot 2^{k/10} \cdot 2^{k/8} \le 2^{-2(m-k) - k /4},
	\] 	
	which is equivalent to
	\[
	\left| N_M / N_D - 2^{-(m-k)} \right| \le 2^{-k/8} \cdot 2^{-(m-k)}.
	\]
\end{proofof}

%\lnote{Above is what I guess should be true.}
\section{A Matching Lower bound}\label{sec:lowerbound}
In this section, we prove that our pseudo-random generator is optimal up to constant factors. That is, we prove that the seed length is optimal. We show that any pseudo-random generator with a seed length of size $s$ can be broken within $O(s)$ rounds (note that our pseudo-random generator is secure up to $\Omega(s)$ rounds when the output of the PRG is of length $n$ for each processor).

%Then, we show that for a pseudo-random generator to be secure for up to $j$ rounds, it must take $O(j)$ rounds to construct it. Our pseudo-random generator takes $10j$ steps to construct, and is secure against all protocols up to $j$ rounds, and hence is optimal up to a constant factor in this sense as well. So, in summary, first we prove that the seed lenght is optimal, and then we prove that the time it takes to construct the PRG is optimal.

\begin{theo}[Seed Length Lower Bound]
Let $n, m$, and $k$ be three integers. Suppose that there are $n$ processors in the Broadcast Congested Clique model. Furthermore, suppose there is a pseudo-random generator which when each processor starts with a seed of size $k$, gives each node a pseudo-random string of size $m$. Then there is an $O(k)$ round Broadcast Congested Clique protocol that can break this PRG.
\end{theo}

\begin{proof}
Consider the following $k+1$-round protocol: each processor broadcasts its first $k+1$ pseudo-random bits. In the case that these bits are pseudo-random, we know that since $nk$ random bits were used as a seed to construct these strings, the transcript of the first round must be one of $2^{nk}$ options. However, in the truly random case, there are $2^{n(k+1)}$ options. So, consider the algorithm that outputs $1$ if the transcript is one of the $2^{nk}$ options consistent with the pseudo-random generator, and otherwise outputs a $0$. Then if the pseudo-random generator was used, then the probability of outputting a $1$ is $1$. In the truly random case, the probability of outputting a $1$ is $\frac{2^{nk}}{2^{n(k+1)}} = \frac{1}{2^n}$. Hence, this algorithm distinguished between the truly random and the pseudo-random case with all but an exponentially small probability.
\end{proof}

%\begin{theo}[Construction Time Lower Bound]
%Let $n, m$, and $k$ be three integers. Suppose that there are $n$ processors in the Broadcast Congested Clique model. Furthermore, suppose there is a pseudo-random generator which when each processor starts with a seed of size $k$, gives each node a pseudo-random string of size $m$ that is constructed within $j$ rounds. Then there is a way to break this PRG within $O(j+k)$ rounds.
%\end{theo}

\section{Discussion}

%One can ask a similar question about the existence of pseudo-random generators for the UCAST Congested Clique (in which instead of broadcasting the same message to all other nodes, a processor may send different messages to different nodes). In \cite{power}, it was shown that the UCAST Congested Clique can simulate strong circuit classes such as TC0. Hence, pseudo-random generators for the UCAST Congested Clique in which every processor only performs efficient computation would imply pseudo-random generators for TC0, which would be a major complexity theoretic result. A UCAST PRG in which the processors are unbounded does not obviously imply any circuit lower bounds, and may still be possible without any major breakthroughs in circuit complexity (though arguably such a result would be less interesting, since it would not be practical).

%Similarly, it may be interesting to construct similar pseudo-random generators for the CONGEST model.

%What other complexity theoretic objects can be translated to these distributed settings? Some ideas: interactive proofs? PCPs?

Our paper leaves some problems open. A main open problem is whether it is possible to improve the planted clique lower bound to show that if the clique is of size $k = \Theta(n^{1/2 - \epsilon})$, the planted clique problem still requires a number of rounds polynomial in $n$.

It would be interesting to extend the framework to work for undirected graphs as well. This causes the rows of the input matrix not to be independent (instead, each pair of rows contain one shared bit). Our current proofs rely on the rows of the input being independent, but we believe it may be possible to extend the framework to also work when the rows exhibit a small amount of dependence.

There are many problems in the $\BCAST$ model that may be interesting to try to prove lower bounds for using the techniques in this paper. These include counting triangles (or $K_4$s) in random graphs, constructing an MST on a complete graph with random weights to the edges, finding communities in a graph sampled from the stochastic block model, the ``planted Hamiltonian cycle'' problem (or, determining whether there is a Hamiltonian cycle in a random graph where the probability of an edge being included is chosen properly so that the probability of such a cycle existing is some constant), graph connectivity, finding the diameter of a random graph (the average degree must be chosen to be low enough so that the diameter is not $2$ with high probability), and APSP on a complete graph with random weight assignments. There are many possibilities.

The uniform distribution is not necessarily the most natural distribution to consider in the broadcast congested clique model. Our techniques are more general and could hopefully be used to prove lower bounds for other distributions as well. It would be interesting to consider other input distributions which may be ``natural" for the problem at hand.

\bibliographystyle{alpha}
\bibliography{team}
\appendix

\section{An Analogue of Newman's Theorem in \BCAST(1)}\label{newman}

In this appendix we adapt Newman's technique from \cite{newman1991private} to show that in the computationally unbounded setting, every randomized $k$ rounds $\BCAST(1)$ protocols in which there are $n$ processors, each with $m$ input bits and outputs $k$ bits at the end, can be simulated with only $O(k \cdot n + \log m)$ public random bits (this has not been observed before this work in the context of the broadcast congested clique). We note that Newman's approach can be adapted to the Unicast Congested Clique model (where each vertex may send different messages to different nodes, instead of broadcasting the same message to all other nodes).

\newcommand{\Pnew}{P_{\sf new}}

Let $\vec{x} = (x_1,x_2,\dotsc,x_n) \in \left(\{0,1\}^{m}\right)^n$ be an input, we use
$\distrP(P,\vec{x})$ to denote the joint distribution of the transcript and the concatenation of all processors' output bits of the $k$-round $\BCAST(1)$ protocol $P$ running on input $\vec{x}$, that is, $\distrP(P,\vec{x})$ is a distribution on $\{0,1\}^{2 k n}$.

We say a protocol $\Pnew$ $\epsilon$-simulates another protocol $P$, if for all possible input $\vec{x}$, we have $\| \distrP(P,\vec{x}) - \distrP(\Pnew,\vec{x}) \| < \epsilon$.

\begin{theo}\label{theo:newman-BCAST}
	Let $P$ be a randomized $\BCAST(1)$ protocol with $n$ processors, each with $m$ input bits and outputs $k$ bits at the end. For all $\epsilon > 0$, there is an equivalent randomized $\BCAST(1)$ protocol $\Pnew$ $\epsilon$-simulating $P$ with only $O(k \cdot n + \log(m) + \log\epsilon^{-1})$ public random bits.
\end{theo}
\begin{proof}
	In the following we are just going to mimic the proof of Newman's theorem.
	
	Without loss of generality we can assume $P$ is a public coin protocol. Suppose it makes use of at most $N$ public coins, where $N$ can be arbitrary large.
	
	Fix an input $x \in \left(\{0,1\}^{m}\right)^n$. Note that $\| \distrP(P,\vec{x}) - \distrP(\Pnew,\vec{x}) \| < \epsilon$ is equivalent to that for all function $f : \{0,1\}^{2 k n} \to \{0,1\}$,
	\[
	\left| \Ex_{p \sim \distrP(P,\vec{x})}[ f(p) ] - \Ex_{p \sim \distrP(\Pnew,\vec{x})}[ f(p) ] \right| < \epsilon.
	\]
	
	We then fix a function $f$. Suppose we draw a public random string $w \sim \distr_{N}$, and we use $P_{w}$ to denote protocol $P$ with public random string setting to $w$ and $P_{w}(\vec{x})$ to denote the concatenation of its transcript and all processor' output bits on input $\vec{x}$.
	
	Now, suppose we pick $T$ $w_1,w_2,\dotsc,w_{T}$ uniform random samples from $\distr_{N}$, by a simple Chernoff bound, we have
	
	\[
	\Pr\left[ \left| \frac{1}{T} \cdot \sum_{i=1}^{T} f(P_{w_i}(\vec{x})) - \Ex_{p \sim \distrP(\Pnew,\vec{x})}[ f(p) ] \right| > \epsilon \right] < \exp(-\Omega(\epsilon^2 \cdot T)).
	\]
	
	Setting
	\[
	T = \Theta(\epsilon^{-2} \cdot \left( nm + 2^{2kn} \right)),
	\]
	it follows
	\[
	\Pr\left[ \left| \frac{1}{T} \cdot \sum_{i=1}^{T} f(P_{w_i}(\vec{x})) - \Ex_{p \sim \distrP(\Pnew,\vec{x})}[ f(p) ] \right| > \epsilon \right] < \frac{1}{10 \cdot 2^{nm} \cdot 2^{2^{2kn}}}.
	\]
	
	Since there are at most $2^{2^{2kn}}$ functions and $2^{nm}$ input bits, by a simple union bound, we have with probability at least $0.9$ over our $T$ samples, $\left| \frac{1}{T} \cdot \sum_{i=1}^{T} f(P_{w_i}(\vec{x})) - \Ex_{p \sim \distrP(\Pnew,\vec{x})}[ f(p) ] \right| < \epsilon$ for all input $\vec{x}$ and function $f : \{0,1\}^{2kn} \to \{0,1\}$.
	
	So we can just pick $T$ samples $w_1,w_2,\dotsc,w_T$ satisfying the above condition, and define $\Pnew$ as the protocol that makes use of $\log T = O(k n + \log m + \log \epsilon^{-1})$ coins to select a random index $i \in [T]$, and act according to $P_{w_i}$. It $\epsilon$-simulates $P$ by the above discussions.
\end{proof}

\begin{rem}
	We remark that in the worst case, at least $\Omega(k \cdot n)$ bits are required to $\epsilon$-simulate a $k$-round \BCAST(1) protocol $P$ where each processor outputs $k$ bits. Since if all processors output $k$ uniform random bits, the total entropy of $\distrP(P,\vec{x})$ on any input $\vec{x}$ is at least $k \cdot n$.
\end{rem}
\section{Algorithm for Planted Clique in $\BCAST(1)$}

In this section we give an algorithm for finding planted clique in $\BCAST(1)$.

\begin{theo}\label{theo:plant-clique-algo}
	Let $n$ be an integer and $ \omega(\log^2 n) \le k \le n$. Given an input from $\distrA_k$, there is an $O(n/k \cdot \polylog(n))$ round $\BCAST(1)$ protocol such that at the end of the protocol, with probability at least $1 - 1/n^2$, all processors know the hidden clique $C$.
\end{theo}
\begin{proof} 
	Let $p = \frac{1}{k} \cdot \log^2 n$.
	
	\paragraph*{Algorithm.}
	The algorithm is very simple. 
	
	\begin{itemize}
		\item At the first round of the protocol, each processor decides to stay active with probability $p$, and broadcasts whether it is active to everyone else. 
		
		\item Let $N_{\sf active}$ be the number of active processors, if $N_{\sf active} > 2 \cdot n \cdot p $, all processors just terminate. 
		
		\item Each active processors broadcast whether it has an edge to each other active processor, which takes $O(n \cdot p) = O(n/k \cdot \polylog(n))$ rounds (i.e., all information about the subrgraph induced by the active processors is broadcasted).
		
		\item Now everyone knows the induced subgraph $G_{\sf active}$ consisting of all active processors. Let the largest clique in $G_{\sf active}$ be $C_{\sf active}$.
		If $|C_{\sf active}| < \frac{1}{2} \cdot \log^2 n$, all processors terminate.
		
		\item Every processor (including the non-active ones) checks whether it is connected to at least a $9/10$ fractions of vertices in $C_{\sf active}$, and if it is, it broadcasts with a message saying it is in the clique $C$. (if it is already in $C_{\sf active}$, then it also says that.)
	\end{itemize}

	\paragraph*{Analysis.} Intuitively, the algorithm works because a random graph doesn't contain a clique of size $10 \log n$ with high probability. And in the hidden clique case, if we pick each vertex with probability $p$, then in expectation we would pick $p \cdot k = \log^2 n$ vertices in $C$, and therefore $|C_{\sf active}| \ge \frac{1}{2} \cdot \log^2 n$ with high probability, while in a random graph the largest clique is of size $\Theta(\log n)$ with high probability.
	
	Let $X_i$ be the random variable indicating whether processor $i$ is active. And let $Y_i$ be the random variable indicating whether processor $i$ is both in the clique and active.
	
	Note that $X_i$'s are i.i.d., by the multiplicative Chernoff bound, we have
	\[
	\Pr\left[N_{\sf active} = \sum_{i=1}^{n} X_i > (1 + \delta) \cdot p \cdot n\right] \le e^{ - \frac{\delta \cdot p \cdot n}{3}}.
	\]
	
	Setting $\delta = 1$, we have with high probability, $N_{\sf active} \le 2 \cdot p \cdot n$.
	
	Note that although $Y_i$'s are not independent, they are negatively associated, and we have the following by another multiplicative Chernoff bound,
	\[
		\Pr\left[\sum_{i=1}^{n} Y_i < (1 - \delta) \cdot p \cdot k\right] \le e^{ - \frac{\delta^2 \cdot p \cdot k}{2}}.
	\]
	
	Set $\delta = 0.5$. With high probability, there are more than $\frac{1}{2} \cdot p \cdot k = \frac{1}{2} \cdot \log^2 n$ active vertices in $C$.
	
	Finally, since with high probability, a random graph doesn't contain a clique of size larger than $10 \log n$. We can conclude that with high probability, at least $\frac{1}{2} \log^2 n - 10\log n$ vertices in $C_{\sf active}$ are actually in $C$. And it is easy to see that the last step of the algorithm identifies the clique $C$ correctly with high probability.
\end{proof}

%\section*{Acknowledgments}

%We would like to thank Shafi Goldwasser, Yang Liu, and Merav Parter for helpful discussions. Many thanks to Sidhanth Mohanty for very fruitful discussions about the planted clique problem.

%Thanks to anonymous reviewers for many useful comments.
\end{document}